\documentclass[reqno]{amsart}
\usepackage{amssymb}
\usepackage{hyperref}

\newtheorem{theorem}{Theorem}[section]
\newtheorem{proposition}[theorem]{Proposition}
\newtheorem{lemma}[theorem]{Lemma}
\newtheorem{corollary}[theorem]{Corollary}

\newtheorem{definition}[theorem]{Definition}

\theoremstyle{remark}
\newtheorem{remark}[theorem]{Remark}

\newtheorem*{note}{Note}

\numberwithin{equation}{section}

\begin{document}

\title[Bethe Ansatz for an open $q$-boson system]
{Completeness of the Bethe Ansatz for an open $q$-boson system with integrable boundary interactions}

\author{J.F.  van Diejen}

\address{
Instituto de Matem\'atica y F\'{\i}sica, Universidad de Talca,
Casilla 747, Talca, Chile}

\email{diejen@inst-mat.utalca.cl}

\author{E. Emsiz}

\address{
Facultad de Matem\'aticas, Pontificia Universidad Cat\'olica de Chile,
Casilla 306, Correo 22, Santiago, Chile}
\email{eemsiz@mat.uc.cl}

\author{I.N. Zurri\'an}

\address{
Facultad de Matem\'aticas, Pontificia Universidad Cat\'olica de Chile,
Casilla 306, Correo 22, Santiago, Chile}

\email{zurrian@famaf.unc.edu.ar}


\thanks{This work was supported in part by the {\em Fondo Nacional de Desarrollo
Cient\'{\i}fico y Tecnol\'ogico (FONDECYT)} Grants \# 1130226,   \# 1141114 and  \# 3160646.}

\date{September 2016}

\begin{abstract}
We employ a discrete integral-reflection representation
of the double affine Hecke algebra of type $C^\vee C$ at the critical level $\text{q}=1$,  to
endow the open finite $q$-boson system with integrable boundary interactions at the lattice ends. It is shown that
the Bethe Ansatz entails a complete basis of eigenfunctions for the commuting quantum integrals in terms of Macdonald's three-parameter hyperoctahedral Hall-Littlewood polynomials.
\end{abstract}

\maketitle



\section{Introduction}
The $q$-boson system \cite{bog-bul:q-deformed,bog-ize-kit:correlation} is an elementary quantum field model built of $q$-deformed oscillators (cf. e.g. \cite{maj:foundations,kli-sch:quantum})
placed on a finite periodic lattice $\mathbb{Z}_m=\mathbb{Z}/m\mathbb{Z}$.  It belongs to a privileged class of one-dimensional particle-conserving quantum field theories for which the $n$-particle Hamiltonian boils down to a (possibly discrete)
quantum integrable Schr\"odinger operator, cf. e.g. \cite{kau:exact,tha:exact,dav-gut:intertwining,kor-bog-ize:quantum,car-lan:loop,lan:second} (and references therein) for some more examples of such integrable quantum field models.
Indeed, in
\cite{kor:cylindric} it was observed that the $q$-boson Hamiltonian acts in the $n$-particle sector of the Fock space as a discrete difference operator that arose in \cite{die:diagonalization} 
as an integrable lattice discretization of the Schr\"odinger operator for the Lieb-Liniger
delta Bose gas on the circle \cite{lie-lin:exact,mat:many-body,dor:orthogonality,kor-bog-ize:quantum}. In both situations, the $n$-particle Bethe Ansatz eigenfunctions had been identified independently as Hall-Littlewood polynomials \cite{tsi:quantum,die:diagonalization,kor:cylindric}. 
Besides the model on the periodic lattice $\mathbb{Z}_m$, analogous versions of the $q$-boson system were considered on the infinite lattice $\mathbb{Z}$ \cite{die-ems:diagonalization} (cf. also \cite{rui:factorized}), on the semi-infinite lattice $\mathbb{N}$ \cite{die-ems:semi-infinite,whe-zin:refined,duv-pas:q-bosons}, and on the open (i.e. aperiodic) finite lattice  \cite{li-wan:exact,die-ems:orthogonality} 
\begin{equation}\label{Nm}
\mathbb{N}_m:= \{ 0,1,2,\ldots ,m\} \qquad (m>0).
\end{equation}
Meanwhile, it has become clear that the $q$-boson system admits a rich variety of generalizations that surfaced naturally  e.g. in  the context  of integrable stochastic  particle processes
\cite{sas-wad:exact,oco-pei:q-weighted,pov:integrability,bor-cor-pet-sas:spectral1,bor-cor-pet-sas:spectral2,tak:deformation} and in the framework of discrete harmonic analysis on Weyl chambers and Weyl alcoves \cite{die:plancherel,die-ems:discrete}.

The purpose of the present work is
to diagonalize an open $q$-boson system on $\mathbb{N}_m$ \eqref{Nm}
endowed with integrable two-parameter boundary interactions at each of the two endpoints $0$ and $m$. Our model is
governed by a (unital associative) $q$-boson field algebra that is deformed at the boundary sites. Specifically, the underlying $q$-boson field algebra is determined by
generators $\beta_l$, $\beta_l^*$ and $q^{N_l}$, $q^{-N_l}$ ($l\in \mathbb{N}_m$) that are ultralocal (i.e. commuting for distinct sites) while satisfying the relations
\begin{subequations}
 \begin{align}
\beta_l q^{N_l} &= q q^{N_l}\beta_l,\  \
q^{N_l} \beta_l^* =q \beta_l^*q^{N_l} ,\ q^{N_l}q^{-N_l}=1,\nonumber \\
 \beta_l\beta_l^*&=(1-c_-\delta_{l}q^{N_0})(1-c_+\delta_{m-l}q^{N_m}) [N_l+1]_q , \label{qbR:a} \\
  \ \beta_l\beta_l^* - q\beta_l^* \beta_l &=(1-c_-\delta_{l}q^{2N_m})(1-c_+\delta_{m-l}q^{2N_m}),\nonumber
\end{align}
where $q\in (-1,1)\setminus \{ 0\}$ and  $c_-,c_+\in (-1,1)$,
\begin{equation}
q^{kN_l}:= (q^{N_l})^k,\qquad [N_l+k]_q :=   \frac{1-q^kq^{N_l}}{1-q} \qquad (l\in\mathbb{N}_m, k\in\mathbb{Z}), \label{qbR:b}
\end{equation}
\end{subequations}
and $\delta_l$ refers to the Kronecker delta (i.e. $\delta_l=1$ if $l=0$ and $\delta_l=0$ otherwise).
The parameters $c_-$ and $c_+$, which deform the $q$-boson field algebra at the sites $l=0$ and $l=m$, are complemented by
two more coupling parameters $g_-, g_+\in \mathbb{R}$ regulating the interactions at these endpoints
via the Hamiltonian 
\begin{equation}\label{Hm}
H :=  g_-[N_0]_q+g_+[N_m]_q +\sum_{0\leq l<m}  \beta_{l+1}^* \beta_l +\beta_l^*\beta_{l+1}.
\end{equation}
For $c_-=c_+=0$ the integrability of $H$ \eqref{Hm} was shown  in \cite{li-wan:exact,die-ems:orthogonality} 
by means of Sklyanin's quantum inverse scattering formalism for open systems with reflecting boundaries \cite{skl:boundary}.  Moreover, for the semi-infinite $q$-boson system on the nonnegative integer lattice (i.e. $m=+\infty$), the integrability of the Hamiltonian in question was inferred for
general (one-sided) boundary parameters $g_-$ and $c_-$ 
\cite{die-ems:semi-infinite,whe-zin:refined}. In each of these cases, the Bethe Ansatz eigenfunctions derived  in \cite{die-ems:semi-infinite,whe-zin:refined,die-ems:orthogonality} turn out to be  given by hyperoctahedral Hall-Littlewood polynomials as introduced by Macdonald \cite{mac:spherical,mac:orthogonal,nel-ram:kostka,par:buildings}.

Below the results 
for the finite open $q$-boson system  will be generalized to the situation of boundary interactions governed by
the full four-parameter family of $c_-,c_+,g_-,g_+$. Rather than to employ quantum inverse scattering as in \cite{li-wan:exact,die-ems:orthogonality}, we recur instead to a representation of Cherednik's double affine Hecke algebra \cite{che:double,mac:affine}
of type $C^\vee C$ \cite{nou:macdonald,sah:nonsymmetric} at the critical level $\text{q}=1$.
The representation in question  is given explicitly in terms of discrete integral-reflection operators. It differs fundamentally from a previous $\text{q}\to 1$  degenerate double affine Hecke algebra representation
\cite{groe:multivariable} in that here no prior rescaling of the parameters and coordinate functions in terms of q is performed.
The particular Hecke-algebraic approach suiting our needs has its origin in the spectral analysis of
Gaudin's generalized Lieb-Liniger delta Bose gas models associated with the (affine) Weyl groups \cite{gau:boundary,gau:bethe,gut:integrable,hec-opd:yang,ems-opd-sto:periodic,ems-opd-sto:trigonometric} by means of integral-reflection operators that were first introduced by  Gutkin and Sutherland \cite{gut-sut:completely}.
The center of the double affine Hecke algebra provides us with the quantum integrals for $H$ \eqref{Hm}, which are subsequently diagonalized in terms of Macdonald's three-parameter hyperoctahedral Hall-Littlewood polynomials via a Bethe Ansatz. From the perspective of double affine Hecke algebras, this entails an extension
of the construction in \cite{die-ems:discrete} enabling to incorporate the important five-parameter master family of type $C^\vee C$. In this picture, our parameters $q$, $c_-$ and $g_-$
are associated with Hecke algebra generators corresponding to
the walls of the Weyl chamber and
parametrize the hyperoctahedral Hall-Littlewood polynomials, while the parameters $c_+$ and $g_+$
are associated with the Hecke algebra generator corresponding to
the affine wall of the Weyl alcove and enter the wave function only through the positions of the Bethe roots. 

The presentation is structured along the following lines. In Section \ref{sec2} the Hamiltonian $H$ \eqref{Hm} is implemented as a self-adjoint operator in the $q$-boson Fock space and its explicit action in the $n$-particle subspace is determined.  
Section \ref{sec3} formulates our main result: the diagonalization of the corresponding $n$-particle Hamiltonian by means of a complete basis of Bethe Ansatz eigenfunctions given by hyperoctahedral Hall-Littlewood polynomials evaluated at the Bethe roots. The bulk of the paper is devoted to the proof of this result within the framework of the double affine Hecke algebra of type $C^\vee C$ at the critical level $\text{q}=1$.
First, the defining properties of this double affine Hecke algebra are detailed in Section
\ref{sec4}.  Then, in Section \ref{sec5}, a concrete representation in terms of discrete integral-reflection operators
is derived by duality from the polynomial representation.
In Section \ref{sec6} the discrete integral-reflection operators are used to construct a Gutkin-Sutherland type propagation operator \cite{gut-sut:completely} that intertwines between the free boson wave functions and our interacting $q$-boson wave functions, respectively.
By computing the image of the free Laplacian on $\mathbb{Z}^n$ with respect to the action of this propagation operator, we
arrive in Section
\ref{sec7} at a quantum integrable deformation of the free Laplacian associated with the double affine Hecke algebra.
The $n$-particle  Hamiltonian for $q$-bosons with open-end boundary interactions is retrieved from this deformed Laplacian in Section
\ref{sec8}  upon symmetrization with respect to the action of the underlying affine hyperoctahedral group.
The propagation operator is then employed once more in Section
\ref{sec9} to show that the Bethe Ansatz eigenfunctions of our $q$-boson Hamiltonian are given by Macdonald's hyperoctahedral Hall-Littlewood polynomials. Finally, the
higher commuting quantum integrals stemming from (the center of) the double affine Hecke algebra at critical level permit to separate the eigenvalue spectrum. This rules out linear dependencies between the eigenfunctions and reduces the proof of
the completeness of the Bethe Ansatz to a straightforward count of the Bethe roots, which are obtained as the minima of a family of strictly convex Morse functions using the classic toolset developed by
Yang and Yang \cite{yan-yan:thermodynamics,gau:bethe,kor-bog-ize:quantum}.

Backup material for some technical subtleties involving our approach is supplied in three appendices at the end. Specifically, in Appendix \ref{appA}
we briefly confirm that  the Poincar\'e-Birkhoff-Witt property and the basic representation of
the double affine Hecke algebra of type $C^\vee C$ \cite{nou:macdonald,sah:nonsymmetric,mac:orthogonal} persist at the critical level $\text{q}=1$, cf. also Ref. \cite{obl:double} for a corresponding verification in the case of the type $A$ double affine Hecke algebra and Ref. \cite{geh:properties} for the verification covering most other types apart from $C^\vee C$. (More precisely, the latter work deals with those types for which the index of the root lattice inside the weight lattice is not equal to $1$.)
In Appendix \ref{appB} we verify certain intertwining relations satisfied by the fundamental propagation operator from Section \ref{sec6}; these relations lie at the basis of the computation that yields the explicit formula for the deformed Laplacian in Section \ref{sec7}.
Finally, in Appendix \ref{appC} we recall Macdonald's product formula for the generalized Poincar\'e series with distinct parameters from \cite{mac:poincare}, in the special case of a
stabilizer subgroup of the affine hyperoctahedral group. This product formula is used  in Section \ref{sec8} to retrieve the open-end $q$-boson Hamiltonian from the double affine Hecke algebra.

\begin{note}
Our notation distinguishes between the parameter $q$ stemming from the $q$-boson algebra and the parameter q associated with the double affine Hecke algebra. Throughout the presentation the value of the latter parameter is assumed to be fixed at the critical level
$\text{q}=1$ (unless explicitly indicated otherwise).
\end{note}

\section{$n$-Particle Hamiltonian}\label{sec2}
Let us define the Fock space 
\begin{subequations}
\begin{equation}\label{Fock:a}
\mathcal{F}:=\bigoplus_{n\geq 0} l^2(\Lambda_{n,m},\Delta_{n,m})
\end{equation}
 consisiting of all series $F=\sum_{n\geq 0} f_n $ with $f_n \in  l^2(\Lambda_{n,m},\Delta_{n,m})$ such that
 \begin{equation}\label{Fock:b}
 (F,F)_m := \sum_{n\geq 0} (f_n,f_n)_{n,m}<\infty .
 \end{equation}
 \end{subequations}
 Here  $l^2(\Lambda_{n,m},\Delta_{n,m})$ stands for the ($n$-particle) Hilbert space of functions
 $f:\Lambda_{n,m}\to\mathbb{C}$ on
\begin{equation}\label{dominant}
\Lambda_{n,m}:=\{\lambda=(\lambda_1,\dots, \lambda_n)\in\mathbb{Z}^n \mid m\geq \lambda_1\geq \lambda_2 \geq \cdots \geq \lambda_n\geq 0\}  ,
\end{equation}
equipped with the following inner product
\begin{subequations}
\begin{equation}\label{inner-product}
(f,g)_{n,m} :=
\sum_{\lambda\in \Lambda_{n,m}} f(\lambda)\overline{g(\lambda)}\Delta_{n,m}(\lambda) 
\end{equation}
associated with the weight function
\begin{equation}\label{weight-function}
\Delta_{n,m}(\lambda)
:=\frac{1}{(c_-;q)_{\text{m}_0(\lambda)}  (c_+;q)_{\text{m}_m(\lambda)}\prod_{l\in\mathbb{N}_m}[\text{m}_l(\lambda)]_q !},
\end{equation}
\end{subequations}
and subject to the additional
convention that $\Lambda_{0,m}:=\{ \emptyset\}$ and $\Delta_{0,m}(\emptyset):=1$ (so $l^2(\Lambda_{0,m},\Delta_{0,m})\cong \mathbb{C}$). The orthogonality measure $\Delta_{n,m}(\lambda)$ \eqref{weight-function}
factorizes in terms of $q$-factorials
\begin{equation*}
[k]_q !:=[k]_q [k-1]_q \cdots [2]_q [1]_q\quad \text{with}\quad [k]_q:=\frac{1-q^k}{1-q}
\end{equation*}
 of the (particle) multiplicities
 \begin{equation*}
\text{m}_l (\lambda) : = |\{ 1\leq j \leq n \mid \lambda_j=l \} | \qquad  (l\in\mathbb{N}_m) ,
\end{equation*}
and is moreover perturbed
at the end-points
by $q$-shifted factorials of the form
\begin{equation*}
(c_\pm;q)_k:=(1-c_\pm)(1-c_\pm q)\cdots (1-c_\pm q^{k-1}),
\end{equation*}
where it is assumed that empty products are equal to $1$ (so $[0]_q !=(c_\pm;q)_0=1$).

We introduce the following actions of the $q$-boson field generators $\beta_l$, $\beta_l^*$ and $q^{\pm N_l}$ ($l\in\mathbb{N}_m$) on
$f\in l^2(\Lambda_{n,m},\Delta_{n,m})\subset \mathcal{F}_m$: 
\begin{subequations}
\begin{equation}\label{qbA:a}
(\beta_l f)(\lambda):=
  f(\beta_l^*\lambda)
\end{equation}
for $\lambda\in \Lambda_{n-1,m}$ if $n>0$ and $\beta_l f:=0$ if $n=0$,
 \begin{align}\label{qbA:b}
(\beta^*_l &  f)(\lambda):= \\
&\begin{cases}
[\text{m}_l(\lambda)]_q (1-c_-\delta_{l}q^{\text{m}_0(\lambda)-1})(1-c_+\delta_{m-l}q^{\text{m}_m(\lambda)-1})f(\beta_l\lambda)&\text{if}\ \text{m}_l(\lambda)>0 \\
0&\text{otherwise}
\end{cases} \nonumber
 \end{align}
for $\lambda\in \Lambda_{n+1,m}$, and
 \begin{equation}\label{qbA:c}
(q^{\pm N_l} f)(\lambda):=q^{\pm \text{m}_l(\lambda)}f(\lambda) 
\end{equation}
\end{subequations}
for $\lambda\in \Lambda_{n,m}$.
Here
$\beta^*_l\lambda\in\Lambda_{n+1,m}$ and $\beta_l\lambda\in \Lambda_{n-1,m}$ denote the partitions obtained from $\lambda\in\Lambda_{n,m}$ by inserting or deleting a part of size $l$, respectively (where  in the latter case it is assumed that $\text{m}_l(\lambda)>0$). 

It is readily verified that the operators in Eqs. \eqref{qbA:a}--\eqref{qbA:c} satisfy the ultralocal $q$-boson algebra relations in Eqs. \eqref{qbR:a}, \eqref{qbR:b}. Hence, we end up with a representation of the $q$-boson algebra on the dense domain $\mathcal{D}_m\subset\mathcal{F}_m$
 \eqref{Fock:a}, \eqref{Fock:b} of \emph{terminating} series $F=\sum_{n\geq 0} f_n $ with $f_n \in  l^2(\Lambda_{n,m},\Delta_{n,m})$.
In this representation, the parts of $\lambda=(\lambda_1,\ldots ,\lambda_n)\in\Lambda_{n,m}$ are interpreted as the positions of $n$ particles---$q$-bosons---on $\mathbb{N}_m$. The operators $\beta_l: l^2(\Lambda_{n,m},\Delta_{n,m})\to  l^2(\Lambda_{n-1,m},\Delta_{n,m})$ and  $\beta_l^*: l^2(\Lambda_{n,m},\Delta_{n,m})\to  l^2(\Lambda_{n+1,m},\Delta_{n,m})$ annihilate and create a $q$-boson at the site $l\in\mathbb{N}_m$, respectively, while $q^{N_l}$ counts the number of $q$-bosons at this site as a power of $q$.
The representation in question is unitary (i.e. it preserves the underlying $\ast$-structure):
\begin{align*}
(\beta_l f,g)_{n,m}&=(f,\beta_l^\ast g)_{n+1,m}\qquad ( f\in l^2(\Lambda_{n+1,m},\Delta_{n+1,m}),\ g\in\ l^2(\Lambda_{n,m},\Delta_{n,m})) ,\\
(q^{\pm N_l} f, g)_{n,m}&=(f,q^{\pm N_l} g)_{n,m}\qquad (f,g\in l^2(\Lambda_{n,m},\Delta_{n,m})) ,
\end{align*}
whereas it is also immediate from the definitions that the $q$-boson  annihilation- and creation operators, as well as  the $q$-deformed number operator, at the site $l\in\mathbb{N
}_m$ are bounded on
$\mathcal{F}_m$:
\begin{align*}
( \beta_l f,\beta_l f )_{n-1,m}&\leq  \frac{(1+|c_-|\delta_l)(1+|c_+|\delta_{m-l})}{1-q} (  f,f)_{n,m}  ,\nonumber\\
\label{bounded} ( \beta_l^* f,\beta_l^*f)_{n+1,m}&\leq \frac{(1+|c_-|\delta_l)(1+|c_+|\delta_{m-l})}{1-q}  (  f,f)_{n,m},\\
( q^{N_l} f,q^{N_l} f )_{n,m}&\leq  ( f,f )_{n,m} , \nonumber
\end{align*}
for any $f\in  l^2(\Lambda_{n,m},\Delta_{n,m}) $.
The upshot is the Hamiltonian $H$ \eqref{Hm} constitutes a bounded self-adjoint on $\mathcal{F}_m$ that preserves the $n$-particle subspace $l^2(\Lambda_{n,m},\Delta_{n,m})$.
From the actions of the $q$-boson generators in Eqs. \eqref{qbA:a}--\eqref{qbA:c}, one readily deduces the following explicit action of $H$ in $l^2(\Lambda_{n,m},\Delta_{n,m})$ (cf. \cite[Prp. 3.1]{die-ems:semi-infinite}
and \cite[Prp. 6.3]{die-ems:orthogonality}).

\begin{proposition}[$n$-Particle Hamiltonian]\label{Hnm:prp} For any
 $f\in l^2(\Lambda_{n,m},\Delta_{n,m})$ and $\lambda\in\Lambda_{n,m}$, one has that
\begin{align}
 (H f)(\lambda)
=& \Bigl( g_-[\emph{m}_0(\lambda)]_q +
g_+[\emph{m}_m(\lambda)]_q \Bigr)
f(\lambda) \\
+&\sum_{\substack{1\leq j \leq n\\ \lambda+e_j\in\Lambda_{n,m}}}
(1-c_-\delta_{\lambda_j} q^{\emph{m}_0(\lambda)-1})
 [\emph{m}_{\lambda_j}(\lambda)]_q f(\lambda+e_j) \nonumber \\
+&\sum_{\substack{1\leq j \leq n\\ \lambda-e_j\in\Lambda_{n,m}}}
(1-c_+\delta_{m-\lambda_j} q^{\emph{m}_m(\lambda)-1})
 [\emph{m}_{\lambda_j}(\lambda)]_q  f(\lambda-e_j) , \nonumber \label{Hnm} 
\end{align}
where $e_1,\ldots ,e_n$ denote the unit vectors of the standard basis in $\mathbb{Z}^n$.
\end{proposition}

\begin{remark}\label{Hnm-q=0:rem}
The actions of $q$-boson generators $\beta_l$, $\beta_l^*$ and $q^{N_l}$ in Eqs. \eqref{qbA:a}--\eqref{qbA:c} and (thus) that of the $n$-particle Hamiltonian in Proposition \ref{Hnm:prp}  are
polynomial in $q$. For $q\to 0$ the $n$-particle Hamiltonian degenerates 
to a self-adjoint Laplacian of the form
\begin{align*}
& (H f)(\lambda)
= \Bigl( g_-\delta_{\lambda_n} +
g_+\delta_{m-\lambda_{1}}\Bigr)
f(\lambda)  +\\
&\sum_{\substack{1\leq j \leq n\\ \lambda+e_j\in\Lambda_{n,m}}}
(1-c_-\delta_{\lambda_j} )^{\delta_{n-j}}
f(\lambda+e_j) 
+\sum_{\substack{1\leq j \leq n\\ \lambda-e_j\in\Lambda_{n,m}}}
(1-c_+\delta_{m-\lambda_j} )^{\delta_{j-1}}
 f(\lambda-e_j) \nonumber
\end{align*} in
$l^2(\Lambda_{n,m},\Delta_{n,m})$ with
$\Delta_{n,m}(\lambda)=(1-c_-)^{-\delta_{\lambda_n}}(1-c_+)^{-\delta_{m-\lambda_1}}$. This discrete Laplacian models $n$ strongly correlated bosons on the finite open lattice $\mathbb{N}_m$ \eqref{Nm} with boundary interactions at the end-points controlled
by the parameters $g_\pm$ and $c_\pm$.  A corresponding system of strongly correlated bosons on the finite periodic lattice $\mathbb{Z}_m$---commonly referred to as the {\em Phase Model}---was studied in detail in
Refs. \cite{bog-ize-kit:correlation,bog:boxed,kor-stro:slnk}.
\end{remark}

\section{Main result: diagonalization}\label{sec3}
Since the dimension 
\begin{equation}
\dim \left( l^2(\Lambda_{n,m},\Delta_{n,m})\right)= \frac{(m+n)!}{m! \, n!}
\end{equation}
of the $n$-particle subspace is finite,  the \emph{existence} of an orthogonal eigenbasis diagonalizing the $n$-particle Hamiltonian in Proposition \ref{Hnm:prp} is immediate from the self-adjointness.  Here we provide an
 \emph{explicit} eigenbasis for the operator in question in terms of Macdonald's hyperoctahedral Hall-Littlewood polynomials (associated with the root system $BC_n$) \cite[\S 10]{mac:orthogonal}:
\begin{subequations}
\begin{align}\label{HLp}
P_\lambda (\boldsymbol{\xi} ;q;a,\hat{a})& :=   \\
 \sum_{\substack{ \sigma\in S_n \\ \epsilon\in \{ 1,-1\}^n}}  & C(\epsilon_1 \xi_{\sigma_1},\ldots , \epsilon_n \xi_{\sigma_n};q,a,\hat{a})
\exp (i\epsilon_1 \xi_{\sigma_1}\lambda_1+\cdots +i \epsilon_n \xi_{\sigma_n} \lambda_n)  \nonumber
\end{align}
with $\boldsymbol{\xi}=(\xi_1,\ldots ,\xi_n)$ and $\lambda\in\Lambda_{n,m}$. The summation is meant over all permutations $\sigma= { \bigl( \begin{smallmatrix}1& 2& \cdots & n \\
 \sigma_1&\sigma_2&\cdots & \sigma_n
 \end{smallmatrix}\bigr)}$ of the symmetric group $S_n$ and all sign configurations
$\epsilon=(\epsilon_1,\ldots,\epsilon_n)\in \{ 1,-1\}^n$, and the expansion coefficients are given explicitly by
\begin{eqnarray}\label{Cp}
\lefteqn{C(\xi_1,\ldots ,\xi_n;q;a,\hat{a}) :=
\prod_{1\leq j\leq n} \frac{(1-a e^{-i\xi_j})( 1-\hat{a} e^{-i\xi_j})}{1-e^{-2i\xi_j}}} && \\
&& \times \prod_{1\leq j<k \leq n} \left(\frac{1-q e^{-i(\xi_{j}-\xi_k)}}{1-e^{-i(\xi_{j}-\xi_k)}}\right)\left(  \frac{1-q e^{-i(\xi_{j}+\xi_k)}}{1-e^{-i(\xi_{j}+\xi_k)}} \right)  .\nonumber
\end{eqnarray}
\end{subequations}

For $\mu\in\Lambda_{n,m}$ \eqref{dominant} and
\begin{equation}\label{domain}
q, a_+,a_-,\hat{a}_+,\hat{a}_- \in (-1,1) ,
\end{equation}
let $\psi_\mu:\Lambda_{n,m}\to\mathbb{C}$ denote the hyperoctahedral Hall-Littlewood function of the form
\begin{equation}\label{HLf}
\psi_\mu (\lambda):= P_\lambda (\boldsymbol{\xi}_\mu ;q;a_-,\hat{a}_- )\qquad (\lambda\in\Lambda_{n,m}),
\end{equation}
where  $\boldsymbol{\xi}_\mu \in\mathbb{R}^n$ represents the unique global minimum  of the semibounded strictly convex Morse function $V^{}_\mu:\mathbb{R}^n\to \mathbb{R}$ given by
\begin{subequations}
\begin{align}\label{q-boson-morse-a}
&V^{}_\mu(\boldsymbol{\xi}):= \sum_{1\le j < k \le n }
\left(
    \int_0^{\xi_j+\xi_k} v_q(u)\text{d}u+   \int_0^{\xi_j-\xi_k} v_q(u)\text{d}u
\right) +  \\
&
\sum_{1\leq j\leq n}  \left(
m\xi_j^2-2\pi (\rho_j+\mu_j)\xi_j 
+   \int^{\xi_j}_0 \bigl(v_{a_-}(u)+v_{\hat{a}_-}(u)+v_{a_+}(u)+v_{\hat{a}_+}(u)\bigr)\text{d}u\right)  ,\nonumber
\end{align}
with $\rho_j:=n+1-j$ ($j=1,\ldots ,n$) and 
\begin{equation}\label{q-boson-morse-b}
v_a(\xi ) := 
\int_0^\xi \frac{ (1-a^2)\ \text{d}u}{1-2a\cos(u)+a^2}
=
i \log
\biggl( \frac{1- ae^{i\xi}}{e^{i\xi} - a }  \biggr) \qquad (-1<a<1).
\end{equation}
\end{subequations}

\begin{theorem}[Diagonalization]\label{diagonal:thm}
For  parameter values belonging to the domain in Eq. \eqref{domain},
the hyperoctahedral Hall-Littlewood functions $\psi_{\mu}$, $\mu\in\Lambda_{n,m}$ are $n$-particle eigenfunctions of the $q$-boson Hamiltonian $H$ \eqref{Hm}:
\begin{subequations}
\begin{equation}\label{EV-eq:a}
H \psi_\mu = E(\boldsymbol{\xi}_\mu) \psi_\mu\quad \text{with}\quad E(\boldsymbol{\xi}):=2\sum_{1\leq j\leq n} \cos (\xi_j) ,
\end{equation}
provided
\begin{equation}\label{EV-eq:b}
c_\pm =a_\pm\hat{a}_\pm \quad \text{and}\quad g_\pm = a_\pm +\hat{a}_\pm .
\end{equation}
\end{subequations}

Moreover, for \emph{nonvanishing} parameter values $q, a_+,a_-,\hat{a}_+,\hat{a}_-$ belonging to the domain in question these eigenfunctions are complete in the sense that they constitute a \emph{basis} for
$ l^2(\Lambda_{n,m},\Delta_{n,m})$.
\end{theorem}

This theorem will arise below as a consequence of a much stronger statement in which all commuting quantum integrals for our open-end $q$-boson model are diagonalized simultaneously.
By invoking the explicit action of $H$ \eqref{Hm} in the $n$-particle subspace given by Proposition \ref{Hnm:prp}, we can reformulate the eigenvalue equation in Theorem \ref{diagonal:thm} as an
affine Pieri rule for the hyperoctahedral Hall-Littlewood polynomials.

\begin{corollary}[Affine Pieri Rule]\label{APR:cor}
For parameters belonging to the domain in Eq. \eqref{domain}, Macdonald's hyperoctahedral Hall-Littlewood polynomials
$P_\lambda (\boldsymbol{\xi};q;a_-,\hat{a}_-)$ with $\lambda\in\Lambda_{n,m}$ satisfy the following affine Pieri formula
\begin{align}\label{APR}
P_\lambda(\boldsymbol{\xi} ;q;a_-,\hat{a}_-)&  \sum_{1\leq j\leq n} (e^{i\xi_j}+e^{-i\xi_j})= \\
& \Bigl( (a_-+\hat{a}_-)[\emph{m}_0(\lambda)]_q +
 (a_+ +\hat{a}_+)[\emph{m}_m(\lambda)]_q \Bigr)
P_\lambda(\boldsymbol{\xi} ;q;a_-,\hat{a}_-) \nonumber \\
+&\sum_{\substack{1\leq j \leq n\\ \lambda+e_j\in\Lambda_{n,m}}}
(1-a_-\hat{a}_-\delta_{\lambda_j} q^{\emph{m}_0(\lambda)-1})
 [\emph{m}_{\lambda_j}(\lambda)]_q P_{\lambda+e_j}(\boldsymbol{\xi} ;q;a_-,\hat{a}_-)  \nonumber \\
+&\sum_{\substack{1\leq j \leq n\\ \lambda-e_j\in\Lambda_{n,m}}}
(1-a_+\hat{a}_+\delta_{m-\lambda_j} q^{\emph{m}_m(\lambda)-1})
 [\emph{m}_{\lambda_j}(\lambda)]_q  P_{\lambda-e_j} (\boldsymbol{\xi} ;q;a_-,\hat{a}_-) \nonumber
\end{align}
at $\boldsymbol{\xi}=\boldsymbol{\xi}_\mu$, $\mu\in\Lambda_{n,m}$.
\end{corollary}

For $\hat{a}_-=\hat{a}_+=0$, Corollary \ref{APR:cor} reproduces an affine Pieri rule found in \cite[Sec. 11.4]{die-ems:orthogonality}.
The
conventional (nonaffine) Pieri rule for Macdonald's hyperoctahedral Hall-Littlewood polynomials \cite[App. A]{die-ems:semi-infinite}, which is valid without quantization restrictions on the values of the polynomial spectral variable $\boldsymbol{\xi}\in\mathbb{R}^n$, is of the form in Eq. \eqref{APR} with  $m\to +\infty$ (so the dependence on $a_+,\hat{a}_+$  drops out).

\begin{remark}
It is expected that the hyperoctahedral Hall-Littlewood functions $\psi_\mu$, $\mu\in\Lambda_{n,m}$ in fact constitute an \emph{orthogonal} eigenbasis for $H$ \eqref{Hm} in $ l^2(\Lambda_{n,m},\Delta_{n,m})$ for all parameter values in the domain \eqref{domain}. When $\hat{a}_-=\hat{a}_+= 0$ (so $g_-= a_-$, $g_+= a_+$ and $c_-= c_+= 0$), this orthogonality was recently confirmed in \cite[Sec. 11.4]{die-ems:orthogonality}.
\end{remark}

\begin{remark}\label{normalization:rem}
The normalization of the  hyperoctahedral Hall-Littlewood polynomials $P_\lambda (\boldsymbol{\xi} ;q;a,\hat{a})$ \eqref{HLp}, \eqref{Cp} is such that the coefficient of the leading monomial  $e^{i\lambda_1\xi_1+\cdots +i\lambda_n\xi_n}$ is given by
$
(a\hat{a};q)_{\text{m}_0(\lambda)}\prod_{l\in\mathbb{N}_m}   [\text{m}_l(\lambda)]_q ! .
$
In particular, for $\lambda=0^n=(0,\ldots ,0)\in\Lambda_{n,m}$, we have that
$
P_{0^n}(\boldsymbol{\xi};q;a,\hat{a})=(a\hat{a};q)_{n}\ [n]_q ! \neq 0.
$
\end{remark}

\begin{remark}\label{critical-point:rem} Since $V_\mu (\boldsymbol{\xi})$   \eqref{q-boson-morse-a}, \eqref{q-boson-morse-b}  is smooth and
$V_\mu(\boldsymbol{\xi})\to +\infty$
for $\boldsymbol{\xi}\to\infty$, 
the existence of a global minimum is guaranteed. The uniqueness of this minimum follows by convexity. Indeed, the Hessian
\begin{align}\label{Hesse}
&H_{j,k}:=\partial_{\xi_j}\partial_{\xi_k} V_\mu (\boldsymbol{\xi}) \\
&={\footnotesize
\begin{cases}
2m + v_{a_-}'(\xi_j) + v_{\hat{a}_-}'(\xi_j)+v_{a_+}'(\xi_j) + v_{\hat{a}_+}'(\xi_j) + \sum_{l\neq j}\bigl (v_q'(\xi_j+\xi_l) +v_q'(\xi_j-\xi_l) \bigr) & \text{if $ k=j$}\\
v_q'(\xi_j+\xi_k) -v_q'(\xi_j-\xi_k) & \text{if $k\neq j$}\\
\end{cases} },\nonumber
\end{align}
is positive definite:
\begin{align*}
\sum_{1\leq j,k\leq n}  x_j x_k H_{j,k}    
= & \sum_{1\leq j\leq n} \Bigl( 2m+  v_{a_-}'(\xi_j) + v_{\hat{a}_-}'(\xi_j)+v_{a_+}'(\xi_j) + v_{\hat{a}_+}'(\xi_j) \Bigr) x_j^2 \\
&+  \sum_{1\leq j<k\leq n} \Bigl( v_q'(\xi_j +\xi_k)(x_j+x_k)^2+
   v_q'(\xi_j -\xi_k)(x_j-x_k)^2\Bigr) \\
 \ge & 2m \sum_{1\leq j\leq n} x_j^2,
\end{align*}
since
$
v^\prime_a(\xi)= \frac{ 1-a^2}{1-2a\cos(\xi)+a^2} >0$ for $a\in (-1,1)$.
From the equation $\nabla_{\boldsymbol{\xi}} V_\mu (\boldsymbol{\xi})=0$ for the corresponding critical point $\boldsymbol{\xi}_\mu$:
\begin{align}\label{q-boson-critical-eq}
2m\xi_j+v_{a_-}(\xi_j)& +v_{\hat{a}_-}(\xi_j)+v_{a_+}(\xi_j)+v_{\hat{a}_+}(\xi_j)\\
 &+\sum_{\substack{1\leq k\leq n\\k\neq j}} \Bigl( v_q(\xi_k+\xi_j)- v_q(\xi_k-\xi_j)\Bigr)=2\pi (\rho_j+\mu_j) 
 \nonumber
\end{align}
($j=1,\ldots,n$), one readily deduces that at $\boldsymbol{\xi}=\boldsymbol{\xi}_\mu$  ($\mu\in\Lambda_{n,m}$)
\begin{subequations}
\begin{equation}\label{e:momentgaps1}
\frac{\pi(\rho_j+\mu_j)}{m+\kappa_-} < \xi_j < \frac{\pi(\rho_j+\mu_j)}{m+\kappa_+}
\end{equation}
(for $1\leq j\leq n$),  and also (by subtracting the $k$th equation from the $j$th equation) 
\begin{equation}\label{e:momentgaps2}
\frac{\pi(\rho_j-\rho_k+\mu_j-\mu_k)}{m+\kappa_-} < \xi_j-\xi_k< \frac{\pi(\rho_j-\rho_k+\mu_j-\mu_k)}{m+\kappa_+} 
\end{equation}
(for $1 \le  j < k \le n$), where
\begin{align}
\kappa_{\pm} := & \frac{(n-1)(1-q^2)}{(1\pm |q|)^2} + \\
&\frac{1}{2}\left( \frac{1-a_-^2}{(1\pm |a_-|)^2}+\frac{1-\hat{a}_-^2}{(1\pm |\hat{a}_-|)^2}+ \frac{1-a_+^2}{(1\pm |a_+|)^2}+\frac{1-\hat{a}_+^2}{(1\pm |\hat{a}_+|)^2} \right).
\nonumber
\end{align}
\end{subequations}
Here one exploits that $v_a(\xi)$ is odd and that $\frac{1-a^2}{(1+|a|)^2} \leq v_a^\prime (\xi)\leq \frac{1-a^2}{(1-|a|)^2} $. Moreover, since $v_a(\xi +2\pi)=v_a(\xi)+2\pi$, it also follows from Eq. \eqref{q-boson-critical-eq} that
$\xi_j<\pi$ at the critical point.
The upshot is that for any $\mu\in\Lambda_{m,n}$, the unique global minimum $\boldsymbol{\xi}_\mu$ of $V_\mu (\boldsymbol{\xi})$   \eqref{q-boson-morse-a}, \eqref{q-boson-morse-b} is assumed inside the open fundamental alcove
\begin{equation}\label{alcove}
 \mathbb{A}:=\{ (\xi_1,\xi_2,\ldots,\xi_n)\in\mathbb{R}^n\mid \pi>\xi_1>\xi_2>\cdots >\xi_n>0\} .
\end{equation}
\end{remark}

\begin{remark}\label{analyticity:rem}
For $\mu\in\Lambda_{n,m}$ and parameters in the domain \eqref{domain}, the dependence of the eigenfunction $\psi_\mu$ \eqref{HLf} on $a_\pm$, $\hat{a}_\pm$ and $q$ is real-analytic.
 Indeed, the critical equation
 \eqref{q-boson-critical-eq}  is real-analytic in this parameter domain and so is the critical point $\boldsymbol{\xi}_\mu$ (by the implicit function theorem, because the Jacobian of the critical equation equals the positive definite Hessian \eqref{Hesse} of  $V_\mu (\boldsymbol{\xi})$ and is thus invertible). Since $P_\lambda (\boldsymbol{\xi};q;a,\hat{a})$ \eqref{HLp}, \eqref{Cp} constitutes a (trigonometric) polynomial in the spectral variable $\boldsymbol{\xi}$ with coefficients that depend polynomially  on the parameters $q$, $a$ and $\hat{a}$, this real-analyticity carries automatically over to the eigenfunction $\psi_\mu $. 
\end{remark}

\begin{remark}  
At $q= 0$  the eigenfunctions in Theorem \ref{diagonal:thm} satisfy the eigenvalue equation for the phase model with open-end boundary interactions governed by the discrete Laplacian in 
Remark \ref{Hnm-q=0:rem} (with $g_\pm=a_\pm+\hat{a}_\pm$, $c_\pm=a_\pm \hat{a}_\pm$ and
$a_\pm,\hat{a}_\pm\in (-1,1)$). 
When $a_\pm=\hat{a}_\pm=0$ the boundary conditions of the Laplacian at the walls $\lambda_j=0$, $\lambda_j=\lambda_k$ and $\lambda_k=m$  ($1\leq j<k\leq n$) become of Dirichlet type. The eigenfunction $\psi_\mu$ \eqref{HLf} 
specializes in this situation to a symplectic Schur function evaluated
at $\boldsymbol{\xi}_\mu=\frac{\pi (\rho+\mu)}{m+n+1}$ (where $\rho:=(\rho_1,\rho_2,\ldots,\rho_n)=(n,n-1,\ldots ,2 ,1)$). 
The analog of the latter phase model on the periodic lattice $\mathbb{Z}_m$ was diagonalized by means of the algebraic Bethe Ansatz in terms of (standard) Schur functions \cite{bog:boxed,kor-stro:slnk}.
In our case, the familiar determinantal structure of the Schur functions turns out to persist at $q= 0$ for
general boundary parameters $a_\pm,\hat{a}_\pm\in (-1,1)$:
\begin{subequations}
\begin{equation}\label{q=1-det:a}
P_\lambda(\boldsymbol{\xi} ;0;a,\hat{a}) = \frac{\det [ p_{n-j+\lambda_j}(\xi_k;a,\hat{a})]_{1\leq j,k\leq n}}{\prod_{1\leq j <k \leq n}   (e^{i\xi_j} +e^{-i\xi_j}-e^{i\xi_k}-e^{-i\xi_k})} ,
\end{equation}
where
\begin{equation}\label{q=1-det:b}
p_\ell (\xi ;a,\hat{a}) :=   \frac{(1-a e^{-i\xi})( 1-\hat{a} e^{-i\xi})}{1-e^{-2i\xi}} e^{i\ell \xi}+  
\frac{(1-a e^{i\xi})( 1-\hat{a} e^{i\xi})}{1-e^{2i\xi}}e^{-i\ell \xi} .
\end{equation}
\end{subequations}
Indeed, when pulling out the overall Vandermonde denominator as in  Eq. \eqref{q=1-det:a} 
 from the $q=0$ specialization of
$P_\lambda(\boldsymbol{\xi} ;q;a,\hat{a}) $ \eqref{HLp}, \eqref{Cp},
one ends up with an alternating sum that coincides manifestly with the expansion of the determinant in the numerator. 
The corresponding $q= 0$ parameter specialization of the affine Pieri formula \eqref{APR} for the polynomials in question becomes
\begin{align}
P_\lambda(\boldsymbol{\xi} ;0;a_-,\hat{a}_-)&  \sum_{1\leq j\leq n} (e^{i\xi_j}+e^{-i\xi_j})= \\
& \Bigl( (a_-+\hat{a}_-)\delta_{\lambda_n}+
 (a_+ +\hat{a}_+)\delta_{m-\lambda_{1}}\Bigr)
P_\lambda(\boldsymbol{\xi} ;0;a_-,\hat{a}_-) \nonumber \\
+&\sum_{\substack{1\leq j \leq n\\ \lambda+e_j\in\Lambda_{n,m}}}
(1-a_-\hat{a}_-\delta_{\lambda_j})^{\delta_{n-j}}
 P_{\lambda+e_j}(\boldsymbol{\xi} ;0;a_-,\hat{a}_-)  \nonumber \\
+&\sum_{\substack{1\leq j \leq n\\ \lambda-e_j\in\Lambda_{n,m}}}
(1-a_+\hat{a}_+\delta_{m-\lambda_j} )^{\delta_{j-1}}
P_{\lambda-e_j} (\boldsymbol{\xi} ;0;a_-,\hat{a}_-) \nonumber
\end{align}
(at $\boldsymbol{\xi}=\boldsymbol{\xi}_\mu$, $\mu\in\Lambda_{n,m}$).
It is not at all clear from Theorem \ref{diagonal:thm} whether the eigenfunctions $\psi_\mu$, $\mu\in\Lambda_{n,m}$ actually remain complete in $l^2(\Lambda_{n,m},\Delta_{n,m})$ at $q=0$, beyond the much weaker and straightforward   completeness for 
\emph{generic} values of the boundary parameters $a_\pm,\hat{a}_\pm\in (-1,1)$. (For generic boundary parameters
the completeness is a priori guaranteed---both at $q=0$ and for generic $q\in (-1,1)$---by
the analyticity in Remark \ref{analyticity:rem} in combination with the elementary fact  that in the Dirichlet situation of vanishing parameters $q,a_\pm,\hat{a}_\pm$ the corresponding symplectic Schur functions form an orthogonal basis for $l^2(\Lambda_{n,m})$, cf. e.g.
\cite[\text{Remark~3.7}]{die-ems:orthogonality-macdonald} with $R=\hat{R}$ of type $C_n$.)
\end{remark}

\section{Double affine Hecke algebra of type $C^\vee C$ at critical level}\label{sec4}
In this section we describe the affine hyperoctahedral group and its associated double affine Hecke algebra \cite{che:double,nou:macdonald,sah:nonsymmetric,mac:affine} at critical level.

\subsection{Affine hyperoctahedral group}
The {\em affine hyperoctahedral group} $W$ is a Coxeter group (of type $\tilde{C}_n$) presented by generators $s_0,\ldots ,s_n$ subject to the relations \cite{hum:reflection}
\begin{equation}\label{W-rel2}
\begin{split}
s_j^2 & =1  \qquad\qquad\qquad  (j=0,\dots, n),\\
s_j s_{j+1} s_j  s_{j+1} &=   s_{j+1} s_j  s_{j+1} s_j \quad  ( j=0\ \text{or}\ j= n-1),\\
s_j s_{j+1} s_j  &= s_{j+1} s_j s_{j+1},\quad \ ( j=1,\dots, n-2),\\
s_ j s_k &= s_k s_j \qquad \qquad\quad  ( |j-k|>1 ).
\end{split}
\end{equation}
A \emph{reduced expression} for $w\in W$ is a decomposition in terms of these generators of the form
\begin{equation}\label{red-exp}
w= s_{j_1}\cdots s_{j_\ell}\qquad ( j_1,\ldots,j_\ell\in \{ 0,\ldots ,n\} )
\end{equation}
such that $\ell $ ($\geq 0$) is minimal. The number generators $\ell $ comprising a reduced expression is referred to as the \emph{length} $\ell(w)$ of  the group element (so $\ell (w)=0$ iff $w=1$).

Given a nonvanishing scale parameter $c\in\mathbb{R}^*:=\mathbb{R}\setminus\{ 0\}$, we consider a faithful action of $W$ on $\mathbb{R}^n$ that is characterized by  \emph{simple reflections}  mapping $x=(x_1,\dots, x_n)\in\mathbb{R}^n$ onto
\begin{align}
s_0(x_1,\dots, x_n) &= (2c-x_1,x_2,\dots, x_n), \nonumber \\
s_j(x_1,\dots, x_n) &= (x_1,\cdots, x_{j-1}, x_{j+1},x_{j},x_{j+2},\cdots, x_n)\qquad (j=1,\ldots ,n-1), \nonumber\\
s_n(x_1,\dots, x_n) &= (x_1,\dots, x_{n-1},-x_n) ,\label{W-action}
\end{align}
respectively.
Below we will always assume that $c$ is integral:
\begin{equation*}
c\in\mathbb{Z}^*:=\mathbb{Z}\setminus\{ 0\}
\end{equation*}
(unless explicitly stated otherwise).

\subsection{Double affine  Hecke algebra}
Let $\tau_0,\ldots,\tau_n$ and $\hat{\tau}_0,\ldots,\hat{\tau}_n$  be parameters in $\mathbb{C}^*:=\mathbb{C}\setminus \{ 0\}$ 
such that
$\tau_j=\hat{\tau}_j=\tau$ for $0<j<n$. \emph{Throughout the paper it will be assumed that none of the five parameters $\tau_0$, $\hat{\tau}_0$, $\tau$,  $\tau_n$, $\hat{\tau}_n$ equals a root of unity.}  We consider the following $\text{q}\to 1$ degeneration of the double affine Hecke algebra of type $C^\vee_n C_n$  \cite{nou:macdonald,sah:nonsymmetric}.

\begin{definition}
The \emph{double affine Hecke algebra $\mathbb H$ of type $C^\vee_n C_n$ at critical level} is
the unital associative algebra over $\mathbb{C}$  with invertible generators $T_0,\dots, T_n$ and commuting 
invertible generators $X_1, \dots, X_n$,  subject to
the \emph{quadratic  relations}
\begin{subequations}
\begin{equation} \label{quadratic-relations}
T_j -T_j^{-1}=\tau_j-\tau_j^{-1}
\qquad (j=0,\ldots,n),
\end{equation}
the \emph{braid relations}
\begin{equation}\label{braid-relations}
\begin{split}
T_j T_{j+1} T_j  T_{j+1} &=   T_{j+1} T_j  T_{j+1} T_j \qquad  ( j=0\ \text{or}\ j= n-1),\\
T_j T_{j+1} T_j  &= T_{j+1} T_j T_{j+1} \qquad \quad (j=1,\dots, n-2) ,\\
T_ j T_k &= T_k T_j \qquad \qquad\qquad   ( |j-k|>1) ,
\end{split}
\end{equation}
and the \emph{cross relations}
\begin{equation} \label{cross-relations}
\begin{split}
T_0 X_1   - X_1^{-1} T_0^{-1}&={\hat{\tau}_0}^{-1}-\hat{\tau}_0,\\
T_n X_n^{-1} - X_n T_n^{-1}&={\hat{\tau}_n}^{-1}-\hat{\tau}_n, \\
T_j X_{j+1} &= X_j T_j^{-1} \qquad\quad  (j=1,\dots, n-1),\\
T_j X_k &= X_k T_j  \qquad (| j-k|>1\ \text{or}\ j=n=k+1).
\end{split}
\end{equation}
\end{subequations}
\end{definition}

For $\lambda=(\lambda_1,\ldots ,\lambda_n)\in\mathbb{Z}^n$ 
and for a reduced expression
$w=s_{j_1}\cdots s_{j_\ell}$ let
\begin{equation*}
X^\lambda:=X_1^{\lambda_1}\cdots X_n^{\lambda_n}\quad\text{and}\quad T_w:= T_{j_1}\cdots T_{j_\ell}
\end{equation*}
(which does not depend on the choice of the reduced expression by virtue of the braid relations).

\begin{proposition}[Poincar\'e-Birkhoff-Witt Property]\label{pbw:prp}
The elements $X^\lambda T_w$  (or alternatively $T_wX^\lambda $), with $\lambda\in\mathbb{Z}^n$ and $w\in W$, form a basis for $\mathbb{H}$ over $\mathbb{C}$.
\end{proposition}

For most types other than $C^\vee C$ a corresponding Poincar\'e-Birkhoff-Witt property was proven
in \cite[\text{Secs. 3 and 5}]{obl:double} and \cite[\text{Sec. 2.1.2}]{geh:properties}. For type $C^\vee C$ with
$\text{q}$ not equal to a root of unity, a proof of the Poincar\'e-Birkhoff-Witt property can be found in
\cite[\text{Sec. 3}]{sah:nonsymmetric}. 
For completeness, we provide a proof  of Proposition \ref{pbw:prp} in Appendix \ref{appA} so as to include our setting of the double affine Hecke algebra of type  $C^\vee C$ at  the critical level $\text{q}=1$.

It follows in particular from the above proposition that the commutative subalgebra of
$\mathbb{H}$ generated by $X_1^{\pm 1},\ldots ,X_n^{\pm 1}$ is isomorphic to the algebra $\mathbb{C}[X]$ of Laurent polynomials in $X_1,\ldots ,X_n$.  Because the integral lattice $\mathbb{Z}^n\subset\mathbb{R}^n$ is stable for the action of the affine hyperoctahedral group $W$ in Eq. \eqref{W-action} (as $c$ is assumed to be integral), we can lift the action
 to $\mathbb{C}[X]$ via
$w (X^\lambda):= X^{w \lambda}$ ($w\in W$). This allows to introduce the corresponding
Demazure-Lusztig operators:
\begin{subequations}
\begin{equation}\label{Tj:a}
 \check{T}_j  := \tau_j s_j + b_j(X)(1-s_j) \qquad (j=0,\ldots ,n),
\end{equation}
where
\begin{equation}\label{Tj:b}
\begin{split}
b_0(X) &= \frac{\tau_0 - \tau_0^{-1} + (\hat{\tau}_0 -\hat{\tau}_0^{-1}) X_1}{1-X_1^2}, \\
b_j(X) &= \frac{\tau-\tau^{-1}}{1-X_j^{-1}X_{j+1}} \qquad ( j=1,\dots, n-1),\\
b_n(X) &= \frac{\tau_n - \tau_n^{-1} + (\hat{\tau}_n -\hat{\tau}_n^{-1}) X_n^{-1}}{1-X_n^{-2}}  .
\end{split}
\end{equation}
\end{subequations}
Since $X^\lambda-X^{s_j\lambda}$ is divisible by the denominator of $b_j(X)$, these Demazure-Lusztig operators  $\check{T}_0,\ldots , \check{T}_n$ \eqref{Tj:a}--\eqref{Tj:b} are well-defined as  linear operators acting on $\mathbb{C}[X]$ (upon interpreting $b_j(X)(X^\lambda-X^{s_j\lambda})$ in terms of the corresponding terminating geometric series, cf. Eq. \eqref{geom-series} below).

\begin{proposition}[Polynomial Representation]\label{pol-rep:prp}
The assignment
$T_j \mapsto  \check{T}_j$ ($j=0,\dots, n$), $X_j \mapsto X_j$ ($j=1,\dots, n$)
extends (uniquely) to a representation  of  $\mathbb H$ on $\mathbb{C}[X]$. 
\end{proposition}
In Appendix \ref{appA}, Proposition \ref{pol-rep:prp} is proven by tweaking a $\text{q}\to 1$ degeneration of  the polynomial representation of the double affine Hecke algebra of type $C^\vee C$ going back to Noumi and Sahi
\cite{nou:macdonald,sah:nonsymmetric}.

\begin{remark}\label{H-relations:rem}
It follows from Proposition \ref{pbw:prp} that $\mathbb{H}$ can be equivalently characterized as the unital associative algebra over $\mathbb{C}$ spanned by the elements
$X^\lambda T_w$  (or alternatively $T_wX^\lambda $) subject to
the relations
\begin{subequations}
\begin{align}
 \label{TwTj}
T_wT_j &=
\begin{cases}
 T_{ws_j}  & \text{if $\ell(ws_j)=\ell(w)+1$},\\
 T_{ws_j} +(\tau_j-\tau_j^{-1})T_w & \text{if $\ell(ws_j)=\ell(w)-1$},
\end{cases} \\
\label{XlambdaXmu} X^\lambda X^\mu &=X^{\lambda +\mu} , \\
\label{cross-relations-2}
T_j X^\lambda - X^{s_j^\prime \lambda }T_j &= b_j(X) (X^\lambda-X^{s'_j\lambda}) ,
 \end{align}
 \end{subequations}
for all $w\in W$, $j\in\{0,\dots,n\}$, and $\lambda,\mu\in\mathbb{Z}^n$. Here $s_j^\prime$ refers to the derivative of the simple reflection $s_j$ (i.e. $s_0^\prime (x_1,\ldots ,x_n)=(-x_1,x_2,\ldots ,x_n)$ and $s_j^\prime =s_j$ for $j=1,\ldots ,n$). The subalgebra $\mathcal{H}\subset\mathbb{H}$ generated by $T_0,\ldots ,T_n$  and spanned by the basis $T_w$, $w\in W$ amounts to Lusztig's three-parameter (viz. $\tau_0$, $\tau$ and $\tau_n$) affine Hecke algebra of type $\tilde{C}_n$  \cite{lus:affine}.
\end{remark}

\section{Lattice integral-reflection operators}\label{sec5}
In this section a representation of $\mathbb{H}$ in terms of lattice integral-reflection operators is derived. 
These operators are discrete analogs---in the spirit of \cite{die-ems:discrete}---of corresponding  integral-reflection operators originating from the Gutkin-Sutherland approach towards the solution of
the spectral problem for Gaudin's generalized Lieb-Liniger models associated with the (affine) Weyl groups
\cite{gau:boundary,gut-sut:completely,gut:integrable,hec-opd:yang,ems-opd-sto:periodic,ems-opd-sto:trigonometric}.
Here we arrive at the pertinent lattice integral-reflection operators  of $C^\vee C$ type by duality from
the polynomial representation in Proposition \ref{pol-rep:prp}.

\subsection{Integral-reflection representation of $\mathbb{H}$}
Let $\mathcal{C}(\mathbb{Z}^n)$ be the space of lattice functions $f:\mathbb{Z}^n\to\mathbb{C}$. The affine hyperoctahedral group acts on these functions via
\begin{equation}
(wf)(\lambda):=f(w^{-1}\lambda)\qquad (w\in W,\, f\in \mathcal{C}(\mathbb{Z}^n),\, \lambda\in\mathbb{Z}^n)
\end{equation}
(where---recall---$c$ is assumed to be a nonzero integer). Upon rewriting the action of the simple reflections in $\mathbb{R}^n$ as
\begin{subequations}
\begin{equation}
s_jx=x-a_j(x)\alpha_j\qquad  (j=0,\ldots ,n), 
\end{equation}
with
\begin{equation}
\alpha_j= 
\begin{cases}  -e_1 &\text{if}\ j=0, \\ e_j-e_{j+1}&\text{if}\ j=1,\ldots ,n-1, \\ e_n&\text{if}\ j=n \end{cases}
\end{equation}
(where $e_1,\ldots ,e_n$ denotes the standard basis of unit vectors in $\mathbb{R}^n$), and
\begin{equation}
a_j(x)= 
\begin{cases}  2(c-x_1) &\text{if}\ j=0, \\ x_j-x_{j+1}&\text{if}\ j=1,\ldots ,n-1, \\ 2x_n&\text{if}\ j=n, \end{cases}
\end{equation}
\end{subequations}
we are in the position to
define corresponding discrete integral-reflection operators $I_j:\mathcal{C}(\mathbb{Z}^n)\to \mathcal{C}(\mathbb{Z}^n) $ of the form
\begin{subequations}
\begin{equation}\label{Ij:a}
I_j:=\tau_j s_j+ J_j\qquad (j=0,\ldots ,n).
\end{equation}
Here $J_j:\mathcal C(\mathbb{Z}^n)\to \mathcal C(\mathbb{Z}^n)$ denotes a discrete weighted integral operator
 that integrates the lattice function
$f\in\mathcal{C}(\mathbb{Z}^n)$ over lattice points on the line segment between  $\lambda$ and $s_j\lambda=\lambda-a_j(\lambda)\alpha_j$:
\begin{equation}\label{Ij:b}
(J_j f)(\lambda ) := 
 \begin{cases}
-\sum_{k=1}^{a_j(\lambda)} u_j(k) f(\lambda-k\alpha_j)   &\text{if}\ a_j(\lambda)> 0 ,\\
0&\text{if}\ a_j(\lambda)=0 ,\\
\sum_{k=0}^{-a_j(\lambda)-1} u_j(k) f(\lambda+k\alpha_j)
&\text{if}\ a_j(\lambda) <0,
\end{cases} 
\end{equation}
where
\begin{equation}\label{Ij:c}
u_j(k) = 
\begin{cases}
\tau_j -\tau_j^{-1} & \text{if $k$ is even}, \\
\hat{\tau}_j -\hat{\tau}_j^{-1} & \text{if $k$ is odd} .
\end{cases}
\end{equation}
\end{subequations}
For any $\mu\in\mathbb{Z}^n$, let us denote by $t_\mu: \mathcal{C}(\mathbb{Z}^n)\to \mathcal{C}(\mathbb{Z}^n)$ the translation operator of the form
\begin{equation}
(t_\mu f)(\lambda) :=f(\lambda-\mu)  \qquad (f\in \mathcal{C}(\mathbb{Z}^n), \ \lambda\in\mathbb{Z}^n).
\end{equation}

The following proposition asserts that the integral-reflection operators $I_0,\ldots ,I_n$ in combination with the unit-translation operators $t_{e_1},\ldots ,t_{e_n}$ give rise to a representation of the type $C^\vee C$ double affine Hecke algebra at critical level on $\mathcal{C}(\mathbb{Z}^n)$.
\begin{proposition}[Integral-Reflection Representation]\label{integral-reflection:prp}
The assignment $T_j\to I_j$ ($j=0,\ldots ,n$),  $X_j\to t_{e_j}$ ($j=1,\ldots ,n$) extends (uniquely)
to a representation of $\mathbb H$ on
$\mathcal{C}(\mathbb{Z}^n)$.
\end{proposition}

\subsection{Proof of Proposition \ref{integral-reflection:prp}}
Let us consider the following nondegenerate bilinear pairing 
$(\cdot,\cdot): \mathcal C(\mathbb{Z}^n) \times \mathbb{C}[X] \to \mathbb{C}$:
\begin{equation}\label{pairing}
(f,p(X)) := ( p(X) f)(0) 
\quad (f\in \mathcal C(\mathbb{Z}^n), \,  p(X) \in \mathbb{C}[X]) ,
\end{equation}
where the action of $p(X)=\sum_\lambda  c_\lambda X^\lambda$ ($c_\lambda\in\mathbb{C}$)  on $\mathcal C(\mathbb{Z}^n)$ is determined by the following action of the basis elements:
$X^\lambda f := t_\lambda f$ ($\lambda\in\mathbb{Z}^n$). So, we have in particular that
$(f,X^\lambda) = (t_\lambda f) (0)=f(-\lambda)$.

To prove the proposition it suffices to verify that for any
$f\in \mathcal C(\mathbb{Z}^n)$ and $\lambda\in\mathbb{Z}^n$:
\begin{subequations}
\begin{align}
(t_{e_j} f, X^\lambda) &= (f, X_j X^{\lambda})\qquad\  (j=1,\ldots ,n),    \label{duality:a} \\
(I_j^{(c)} f, X^\lambda ) &= (f, \check{T}_j^{(-c)}X^\lambda )\quad\ (j=0,\ldots ,n), \label{duality:b} 
\end{align}
\end{subequations}
where the superscripts indicate that opposite values for $c$ have to be chosen in the difference-reflection representation and the polynomial representation. (Notice in this connection that  the actions of $s_j$ and thus that of $I_j$ and $\check{T}_j$ only depend on $c$ for $j=0$.)
By acting on an arbitrary basis element $X^\lambda$ with both sides of the quadratic relations, the braid relations and the cross relations for 
$\check{T}_0^{(-c)},\check{T}_1,\ldots ,\check{T}_n$ and $X_1,\ldots ,X_n$, one readily verifies the corresponding relations for 
$I_0^{(c)}, I_1,\ldots ,I_n$ and $t_{e_1},\ldots,t_{e_n}$ upon pairing with $f\in\mathcal{C}(\mathbb{Z}^n)$ and using  Eqs. \eqref{duality:a}, \eqref{duality:b}.
Here it is exploited that the pairing $(\cdot ,\cdot)$ \eqref{pairing} is nondegenerate.

While Eq. \eqref{duality:a} is an immediate consequence of the above definitions, Eq. \eqref{duality:b} follows similarly from the explicit actions of $I_j$ \eqref{Ij:a}--\eqref{Ij:c} and $\check{T}_j$ \eqref{Tj:a}--\eqref{Tj:b} upon invoking the geometric series expansion
\begin{equation}\label{geom-series}
b_j(X)(X^\lambda-X^{s_j\lambda})=
 \begin{cases}
\sum_{k=0}^{a_j(\lambda)-1} u_j(k) X^{\lambda-k\alpha_j}  &\text{if}\ a_j(\lambda)> 0 ,\\
0&\text{if}\ a_j(\lambda)=0 ,\\
-\sum_{k=1}^{-a_j(\lambda)} u_j(k) X^{\lambda+k\alpha_j}
&\text{if}\ a_j(\lambda) <0,
\end{cases} 
\end{equation}
for $j=0,\ldots ,n$.

\section{Lattice propagation operator}\label{sec6}
One of the principal tools to construct the Bethe-Ansatz eigenfunctions for the open $q$-boson Hamiltonian $H$ \eqref{Hm} is provided by a propagation operator stemming from the integral-reflection representation of $\mathbb{H}$. For the graded affine Hecke algebra and its double affine counterpart at critical level such propagation operators were employed in \cite{gut-sut:completely,gut:integrable,hec-opd:yang} and \cite{ems-opd-sto:periodic,ems-opd-sto:trigonometric}, respectively,
to construct the Bethe-Ansatz wave functions for Gaudin's generalized Lieb-Liniger delta Bose gas models associated with the (affine) Weyl groups \cite{gau:boundary,gau:bethe}.
The propagation operator relevant for our present purposes turns out to be the $C^\vee C$-type analog of a lattice propagation operator introduced in \cite{die-ems:discrete}.

\subsection{Propagation operator} From now on we set
\begin{equation}\label{c=m}
\boxed{c=m\quad\text{with}\quad m>0.}
\end{equation}
A fundamental domain for the action of $W$ in $\mathbb{Z}^n\subset \mathbb{R}^n$ (cf. Eq. \eqref{W-action}) is then given by the \emph{fundamental alcove}
$\Lambda_{n,m}$ \eqref{dominant}. For $\lambda\in\mathbb{Z}^n$, let $w_\lambda\in W$ denote the (unique) shortest group element such that
\begin{equation}
\lambda_+:=w_\lambda \lambda \in \Lambda_{n,m}.
\end{equation}
Moreover,
for $w\in W$ let $I_w:\mathcal{C}(\mathbb{Z}^n)\to\mathcal{C}(\mathbb{Z}^n)$ and $\tau_w\in\mathbb{C}^*$ 
be the respective images of $T_w\in \mathcal{H}\subset\mathbb{H}$
with respect to the lattice integral-reflection representation in Proposition \ref{integral-reflection:prp} and the trivial representation of $\mathcal{H}$ onto $\mathbb{C}$  determined by the assignment $T_j\to\tau_j$, $j=0,\ldots ,n$
(cf. Remark \ref{H-relations:rem}).

The \emph{propagation operator} $\mathcal{J}:\mathcal{C}(\mathbb{Z}^n)\to\mathcal{C}(\mathbb{Z}^n)$ is now defined by the following linear action
on $f\in \mathcal{C}(\mathbb{Z}^n)$:
\begin{equation}\label{propagator}
(\mathcal Jf)(\lambda):=\tau_{w_\lambda}^{-1} (I_{w_\lambda} f)(\lambda_+) \qquad (\lambda\in \mathbb{Z}^n) .
\end{equation}
Clearly this propagation operator acts trivially on lattice functions supported inside the fundamental alcove:   $ (\mathcal Jf)(\lambda)=f(\lambda)$ for $\lambda\in \Lambda_{n,m}$.

\begin{proposition}[Bijectivity]\label{inv:prp}
The propagation operator $\mathcal{J}$  \eqref{propagator} is bijective (i.e. it constitutes a linear automorphism of $\mathcal{C}(\mathbb{Z}^n)$).
\end{proposition}

\subsection{Proof of Proposition \ref{inv:prp}}\label{inv:prf}
To $\lambda\in\mathbb{Z}^n\subset\mathbb{R}^n$ we attach the following convex polytope:
\begin{equation}
[\lambda]:= \text{Conv} \{ v^{-1} \lambda_+\mid v\leq w_\lambda\}\subset\mathbb{R}^n, 
\end{equation}
where \emph{Conv} refers to the \emph{convex hull} and the comparison of the group elements is meant with respect to the \emph{Bruhat partial order} on $W$, i.e. $v\leq w$ iff a reduced expression for $v$ can be obtained from a reduced expression for $w$ by deleting simple reflections \cite{hum:reflection}. These polytopes (which degenerate to a point if $\ell (w_\lambda)=0$ and to a line segment if   $\ell (w_\lambda)=1$)  give rise to the following inclusion partial order $\preceq$ on $\mathbb{Z}^n$:
\begin{equation}\label{Z-order}
\forall\mu,\lambda\in\mathbb{Z}^n:\quad \mu\preceq\lambda \Leftrightarrow [\mu]\subseteq [\lambda] .
\end{equation}

\begin{lemma}[Triangularity]\label{triangular:lem}
The propagation operator $\mathcal{J}$ \eqref{propagator} is triangular with respect to the inclusion partial order \eqref{Z-order}, viz.
\begin{equation}\label{triangular}
\forall f\in \mathcal{C}(\mathbb{Z}^n),\, \lambda\in \mathbb{Z}^n:\qquad
(\mathcal{J}f)(\lambda )= \sum_{\mu\in \mathbb{Z}^n,\, \mu \preceq \lambda} J_{\lambda,\mu} f(\mu),
\end{equation}
for certain expansion coefficients $J_{\lambda,\mu}\in\mathbb{C}$  with $J_{\lambda ,\lambda}=\tau_{w_\lambda}^{-2}$ ($\neq 0$).
\end{lemma}

\begin{proof}
The proof proceeds by induction with respect to the length of $w_\lambda$.
For $\ell (w_\lambda )=0$ (i.e. $\lambda \in \Lambda_{n,m}$),
our polytope $[\lambda]$ degenerates to the single point $\lambda $ and
the stated triangularity becomes trivial: $(\mathcal{J}f)(\lambda) =f(\lambda)$ (as noticed above just after Eq. \eqref{propagator}). For $\ell(w_\lambda)>0$ let us pick $j\in\{ 0,\ldots ,n\}$ such that $w_\lambda s_j <w_\lambda$, i.e. $w_\lambda =w_{s_j\lambda}s_j$ with $\ell (w_\lambda)=\ell (w_{s_j\lambda})+1$ (so $w_{s_j\lambda}< w_\lambda$  and $s_j\lambda\prec \lambda$).
One then has that
\begin{align*}
(\mathcal{J}f)(\lambda) =& \tau_{w_\lambda}^{-1}(I_{w_{\lambda}} f)(\lambda_+)=
\tau_j^{-1}\tau_{w_{s_j\lambda}}^{-1}
(I_{w_{s_j\lambda}} I_{j} f)\bigl((s_j\lambda )_+\bigr) \\
 \stackrel{\text{(i)}}{=}&
\tau_j^{-1}\sum_{ \mu\in \mathbb{Z}^n,\,\mu \preceq s_j\lambda} J_{s_j\lambda ,\mu} (I_{j} f) (\mu )
\stackrel{\text{(ii)}}{=}
\sum_{\mu\in \mathbb{Z}^n,\, \mu \preceq \lambda} J_{\lambda ,\mu} f(\mu ),
\nonumber
\end{align*}
where step $\text{(i)}$ hinges on the induction hypothesis and in step $\text{(ii)}$ it was used that---while
$(I_jf)(\mu )$ involves evaluations of $f$ at lattice points on the line segment joining
$\mu$ and $s_j\mu$---both the polytopes
$[s_j\lambda]$ and $s_j[s_j\lambda]$ are contained in the polytope $ [\lambda]$. Finally,
upon invoking the explicit action of the lattice integral-reflection operator $I_j$ \eqref{Ij:a}--\eqref{Ij:c} and comparing the leading coefficients on both sides of the equality of step $\text{(ii)}$, it is readily seen that $J_{\lambda,\lambda}=\tau_j^{-2}J_{s_j\lambda,s_j\lambda}$, whence $J_{\lambda,\lambda}=\tau_{w_\lambda}^{-2}$ (again by induction).
\end{proof}

The triangularity in Lemma \ref{triangular:lem} implies that for a given  $g\in\mathcal{C}(\mathbb{Z}^n)$, the value
of $f\in\mathcal{C}(\mathbb{Z}^n)$ at any point $\lambda\in\mathbb{Z}^n$ can be
uniquely solved from the linear equation $(\mathcal{J}f)(\lambda)=g(\lambda)$ by performing induction with respect to the inclusion order  \eqref{Z-order}. 
Hence, the lattice propagation operator $\mathcal{J}:\mathcal{C}(\mathbb{Z}^n)\to \mathcal{C}(\mathbb{Z}^n)$ is  bijective.

\section{Deformed Laplacian in $\mathcal{C}(\mathbb{Z}^n)$}\label{sec7}
In this section the (invertible) propagation operator $\mathcal{J}$ \eqref{propagator} is employed to construct an integrable Laplacian  in $\mathcal{C}(\mathbb{Z}^n)$
associated with $\mathbb{H}$.

\subsection{Integrability}\label{sec7.1}
The hyperoctahedral group $W_0\subset W$ arises as the finite subgroup of $W$ generated by $s_1,\ldots ,s_n$. In the explicit representation of Eq. \eqref{W-action}, this subgroup acts on the coordinates of $x\in\mathbb{R}^n$ as the group of \emph{signed permutations}. The action of $w\in W$ decomposes in turn as $w=v t_{\mu}=t_{v\mu}v$ with $v\in W_0$ and $\mu\in 2m\mathbb{Z}^n$ (cf. Eq. \eqref{c=m}), where $t_\mu x:=x+\mu$. The derivative 
$w^\prime\in W_0$ ignores the affine translation: $w^\prime=(vt_\mu)^\prime:=v$
(cf. Remark \ref{H-relations:rem}).

It is evident from the cross relations in Eq. \eqref{cross-relations-2} that the $W_0$-invariant subalgebra $\mathbb{C}[X]^{W_0}:=\{ p\in \mathbb{C}[X] \mid wp=p,\,\forall w\in W_0\}$ belongs to the center $\mathcal{Z}(\mathbb{H}):=\{ z\in\mathbb{H} \mid zh=hz,\, \forall h\in\mathbb{H} \}$ of the double affine Hecke algebra at critical level:
\begin{equation}\label{center}
\mathbb{C}[X]^{W_0}\subset \mathcal{Z}(\mathbb{H}).
\end{equation}
The elementary symmetric functions
\begin{equation}\label{elementary-sf}
E_r (X_1,\ldots ,X_n):= \sum_{\substack{ J\subset \{ 1,\ldots ,n\} \\ | J| =r}} \prod_{j\in J} (X_j+X_j^{-1})\qquad (r=1,\ldots ,n),
\end{equation}
provide a system of algebraically independent generators for $\mathbb{C}[X]^{W_0}$. (Here $|J|$ refers to the number of elements of $J\subset \{ 1,\ldots ,n\} $.) By conjugating the images of these elementary symmetric functions in the integral-reflection representation of Proposition \ref{integral-reflection:prp} with respect to the propagation operator $\mathcal{J}$ \eqref{propagator}, we arrive
at commuting operators $L_1,\ldots ,L_n$ representing a
 quantum integrable system in $\mathcal{C}(\mathbb{Z}^n)$:
 \begin{subequations}
\begin{equation}\label{Lr}
L_r :=\mathcal J E_r(t) \mathcal J^{-1} \qquad (r=1,\ldots ,n),
\end{equation}
where
\begin{equation}\label{Er}
E_r(t) := E_r (t_{e_1},\ldots ,t_{e_n})= \sum_{\mu\in W_0(e_1+\cdots +e_r)}    t_\mu  .
\end{equation}
\end{subequations}

\subsection{Laplacian associated with $\mathbb{H}$} 
The simplest of the above quantum integrals $L_1=\mathcal{J}E_1(t)\mathcal{J}^{-1}$ with $E_1(t)=\sum_{1\leq j\leq n}( t_{e_j}+t_{-e_j})$ stems from the first elementary symmetric function $E_1(X_1,\ldots ,X_n)=\sum_{1\leq j\leq n} (X_j+X_j^{-1})$.
The following proposition reveals that  the operator in question acts in $\mathcal{C}(\mathbb{Z}^n)$ as a deformed Laplacian.

\begin{proposition}[Deformed Laplacian]\label{laplacian:prp}
The explicit action of
$L:=L_1$ \eqref{Lr}, \eqref{Lr}  on $\mathcal{C}(\mathbb{Z}^n)$ is of the form
\begin{subequations}
\begin{align}\label{La}
(L f)(\lambda)=
\sum_{1\leq j\leq n } \Bigl( \tau_{w_{w_\lambda (\lambda+ e_j)}}^{2}  f(\lambda +e_j) &+ \tau_{w_{w_\lambda (\lambda- e_j)}}^{2}  f(\lambda -e_j) \\
&+ \bigl( d_{\lambda_+ , e_j}+  d_{\lambda_+ ,- e_j}  \bigr) f(\lambda) \Bigr) \nonumber
\end{align}
($f\in \mathcal{C}(\mathbb{Z}^n)$, $\lambda\in \mathbb{Z}^n$),
with
\begin{equation} \label{Lb}
d_{\lambda,\nu}
:=
\begin{cases}
\tau^{2(j-1)}\tau_0(\hat{\tau}_0-\hat{\tau}_0^{-1}) & \text{if}\ \lambda_j=m\ \text{and}\  \nu=e_j,  \\
 \tau^{2(n-j)}\tau_n(\hat{\tau}_n-\hat{\tau}_n^{-1})& \text{if}\  \lambda_j=0\  \text{and}\ \nu=-e_j, \\
0 & \text{otherwise}.
\end{cases}
\end{equation}
\end{subequations}
\end{proposition}

\begin{proof}
Given an arbitrary lattice function $f:\mathbb{Z}^n\to\mathbb{C}$, let $g:=\mathcal{J}^{-1}f\in\mathcal{C}(\mathbb{Z}^n)$.  One has that
for any $\lambda\in \mathbb{Z}^n$:
\begin{align*}
  (Lf)(\lambda) &= (\mathcal{J} E_1(t) g) (\lambda)\stackrel{\text{Eq.}~\eqref{propagator}}{=}
\tau_{w_\lambda}^{-1} (I_{w_\lambda} E_1(t) g)(\lambda_+)  \\
&\stackrel{\text{Eq.}~\eqref{center}}{=}\tau_{w_\lambda}^{-1} (E_1(t) I_{w_\lambda} g)(\lambda_+)
\stackrel{\text{(i)}}{=} \tau_{w_\lambda}^{-1} \sum_{\nu\in W_0e_1} (I_{w_\lambda} g)(w_\lambda(\lambda+\nu))\\
&\stackrel{\text{(ii)}}{=}
\sum_{\nu\in W_0e_1}   \Bigl( \tau_{w_{w_\lambda (\lambda+\nu)}}^{2} f (\lambda+\nu)+d_{\lambda_+,w_\lambda'\nu} f (\lambda) \Bigr) .
\end{align*}
Here we relied (i) on Eq.~\eqref{Er} and the elementary property that $w\lambda +W_0\mu=w(\lambda+W_0\mu)$ for any $\lambda, \mu\in\mathbb{Z}^n$  and $w\in W$, and (ii) on the affine intertwining relation in
Eq.~\eqref{intertwining-property} of Appendix \ref{appB}.
\end{proof}

\section{Quantum integrability of the open $q$-boson system}\label{sec8}
By pushing the integral-reflection representation restricted to $\mathcal{H}\subset\mathbb{H}$ through the propagation operator $\mathcal{J}$ \eqref{propagator}, we arrive at a difference-reflection representation of the affine Hecke algebra of type $\tilde{C}_n$ on $\mathcal{C}(\mathbb{Z}^n)$ (cf. Remark \ref{H-relations:rem}). This difference-reflection representation is subsequently used to retrieve
a system of $n$ commuting quantum integrals for the
$n$-particle $q$-boson Hamiltonian in Proposition \ref{Hnm:prp} as the $W$-invariant reduction of the commuting quantum integrals for the deformed Laplacian $L$ \eqref{La}, \eqref{Lb}.

\subsection{Difference-reflection representation of $\mathcal{H}$}
For $j\in \{0,\ldots,n\}$,  let  $\hat{T}_j:\mathcal{C}(\mathbb{Z}^n)\to \mathcal{C}(\mathbb{Z}^n)$ be the difference-reflection operator of the form
\begin{equation}\label{That}
(\hat T_j f)(\lambda):=\tau_j f(\lambda) + \tau_j^{\text{sgn}(a_j(\lambda))}\bigl( f(s_j\lambda)-f(\lambda) \bigr)
\end{equation}
$(f\in \mathcal{C}(\mathbb{Z}^n),\ \lambda\in\mathbb{Z}^n)$, where $\text{sgn}(x):=1$ if $x\geq 0$ and $\text{sgn}(x):=-1$ if $x< 0$.

\begin{proposition}[Intertwining Relations]\label{IT:prp}
The (invertible) propagation operator $\mathcal{J}$ \eqref{propagator} intertwines between the integral-reflection operators and the difference-reflection operators:
\begin{equation}
\mathcal{J} I_j =  \hat{T}_j  \mathcal{J} \qquad (j=0,\ldots ,n). \label{IT:relations} 
\end{equation}
\end{proposition}

\begin{proof}
From the definitions it is seen that
for any $j\in \{0,\ldots n\}$, $f\in\mathcal{C}(\mathbb{Z}^n)$ and $\lambda\in\mathbb{Z}^n$:
\begin{align*}
(\mathcal{J} I_j f)(\lambda ) \stackrel{\text{(i)}}{=}& \tau^{-1}_{w_\lambda}  ( I_{w_\lambda} I_j f)( \lambda_+ )  \\
\stackrel{\text{(ii)}}{=} &
\begin{cases}
\tau^{-1}_{w_\lambda}
  (I_{w_\lambda s_j} f)( \lambda_+ ) &\text{if}\   a_j (\lambda)\geq 0 \\
\tau^{-1}_{w_\lambda}
 \left( (I_{w_\lambda s_j} f)( \lambda_+ )
 +(\tau_j-\tau_j^{-1})
 ( I_{w_\lambda} f)( \lambda_+ )  \right) &\text{if}\   a_j (\lambda)< 0
 \end{cases} \\
 =  &\tau_j \tau^{-1}_{w_\lambda} (I_{w_\lambda} f)(\lambda_+ ) \\
& +
\tau_j^{\text{sgn} (a_j(\lambda ))}
\left(
\tau^{-1}_{w_\lambda s_j} (I_{w_\lambda s_j} f)( \lambda_+ ) -
\tau^{-1}_{w_\lambda} (I_{w_\lambda} f)( \lambda_+ )
\right) \\
\stackrel{\text{(iii)}}{=} &  (\hat{T}_j\mathcal{J}f)(\lambda ) ,
\end{align*}
where we used (i) Eq.~\eqref{propagator}, (ii) Eq.~\eqref{TwTj} and the relation $\ell (w_\lambda s_j)=\ell (w_\lambda )+\text{sgn}(a_j(\lambda))$, and (iii) Eqs.~\eqref{propagator}, \eqref{That}
and the observation that 
\begin{equation*}
\tau^{-1}_{w_\lambda s_j} (I_{w_\lambda s_j} f)( \lambda_+ )=\tau^{-1}_{w_{ s_j\lambda}} (I_{w_{ s_j\lambda}} f)( \lambda_+ ) .
\end{equation*}
\end{proof}

It is immediate from the intertwining relations in Eq.  \eqref{IT:relations} (and the invertibility of $\mathcal{J}$ guaranteed by Proposition \ref{inv:prp})  that the difference-reflection operators $\hat{T}_0,\ldots ,\hat{T}_n$  inherit the quadratic relations and the braid relations for the affine Hecke algebra generators  from the corresponding relations satisfied by the integral-reflection operators $I_0,\ldots ,I_n$ (cf. Proposition \ref{integral-reflection:prp}).

\begin{corollary}[Difference-Reflection Representation]\label{difference-reflection:cor}
The assignment $T_j\to \hat{T}_j$ ($j=0,\ldots ,n$) extends (uniquely)
to a representation of $\mathcal{H}$ on
$\mathcal{C}(\mathbb{Z}^n)$.
\end{corollary}

\subsection{$W$-Invariant reduction}
The $W$-invariant subspace
\begin{subequations}
\begin{eqnarray}\label{W-inv-subspace:a}
\mathcal{C}(\mathbb{Z}^n)^{W} &:=& \{ f\in\mathcal{C}(\mathbb{Z}^n)\mid  s_j f=f,\ j=0,\ldots ,n\} \\
&\stackrel{\text{Eq.~\eqref{That}}}{=}& \{ f\in\mathcal{C}(\mathbb{Z}^n)\mid  \hat{T}_j f=\tau_j f,\ j=0,\ldots ,n\}
\label{W-inv-subspace:b}
\end{eqnarray}
\end{subequations}
consists of the lattice functions $f:\mathcal{C}(\mathbb{Z}^n)\to\mathbb{C}$ that are permutation-invariant, even, and periodic with period $2m$ in the coordinates  $\lambda_1,\ldots ,\lambda_n$ of the variable $\lambda\in\mathbb{Z}^n$. 
This subspace turns out to be stable with respect to the action of $L_r$
 \eqref{Lr}, \eqref{Er}.

\begin{proposition}[$W$-Invariant Reduction]\label{W-reduction:prp}
The commuting quantum integrals $L_1,\ldots ,L_n$ \eqref{Lr}, \eqref{Er} map the $W$-invariant subspace $\mathcal{C}(\mathbb{Z}^n)^{W}$  into itself.
\end{proposition}
\begin{proof}
For any $f\in \mathcal{C}(\mathbb{Z}^n)^W$ and $ r\in \{ 1,\ldots ,n\}$, one has that
\begin{align*}
\hat{T}_j L_rf & \stackrel{\text{Eq.~\eqref{Lr}}}{=}  \hat{T}_j \mathcal{J} E_r(t) \mathcal{J}^{-1} f \stackrel{\text{Eq.~\eqref{IT:relations}}}{=}\mathcal{J} I_j E_r(t) \mathcal{J}^{-1} f   \\ & \stackrel{\text{Eq.~\eqref{center}}}{=}
\mathcal{J} E_r(t) I_j  \mathcal{J}^{-1} f \stackrel{\text{Eq.~\eqref{IT:relations}}}{=}  \mathcal{J} E_r(t)  \mathcal{J}^{-1} \hat{T}_j f \stackrel{\text{Eq.~\eqref{W-inv-subspace:b}}}{=} 
\tau_j L_rf
\end{align*}
for $j=0,\ldots ,n$, whence $L_rf \in \mathcal{C}(\mathbb{Z}^n)^W$.
\end{proof}

\subsection{Quantum integrals for the $q$-boson Hamiltonian}
Let
$\Pi : \mathcal{C}(\mathbb{Z}^n)^W\to\mathcal{C}(\Lambda_{n,m})$   
be the linear isomorphism of the form
\begin{subequations}
\begin{equation}
(\Pi f)(\lambda) :=f(\lambda )\quad (f\in  \mathcal{C}(\mathbb{Z}^n)^W,\, \lambda\in\Lambda_{n,m}),
\end{equation}
with the inverse map given by
\begin{equation}
(\Pi^{-1} f)(\lambda) = f (\lambda_+) \quad (f\in  \mathcal{C}(\Lambda_{n,m}),\, \lambda\in\mathbb{Z}^n) .
\end{equation}
\end{subequations}
This explicit isomorphism between  $\mathcal{C}(\mathbb{Z}^n)^W$ and $\mathcal{C}(\Lambda_{n,m})$
allows us to interpret the $W$-invariant reductions of the operators $L_1,\ldots ,L_n$ in Proposition \ref{W-reduction:prp}
as commuting operators
$H_1,\ldots ,H_n$ on $\mathcal{C}(\Lambda_{n,m})$, where
\begin{equation}\label{q-boson-integrals}
H_r :=    \Pi \, L_r \,\Pi^{-1} \qquad (r=1,\ldots ,n).
\end{equation}
For $r=1$ the action of $H_r$ in $\mathcal C(\Lambda_{n,m})$ can be made explicit using Proposition \ref{laplacian:prp}.

\begin{proposition}[$q$-Boson Hamiltonian]\label{H:prp}
The explicit action of  $H:=H_1$  \eqref{q-boson-integrals} on $f\in \mathcal C(\Lambda_{n,m})$ is given by
\begin{subequations}
\begin{equation}\label{Ha}
( H f)(\lambda)
= u(\lambda) f(\lambda) + \sum_{\substack{1\le j \le n,\, \epsilon\in \{1,-1\} \\ \lambda +\epsilon e_j \in \Lambda_{n,m}}} v_j(\lambda) f(\lambda+\epsilon e_j)
\qquad ( \lambda\in \Lambda_{n,m} ),
\end{equation}
where
\begin{equation}\label{Hb}
\begin{split}
v_j(\lambda) &:=
 [\emph{m}_{\lambda_j}(\lambda)]_{\tau^2} 
   (1+\delta_{\lambda_j} \tau_n^2 \tau^{2(\emph{m}_0(\lambda)-1)})
   (1+\delta_{m-\lambda_j} \tau_0^2 \tau^{2(\emph{m}_m(\lambda)-1)})
\end{split}
\end{equation}
and
\begin{equation}  \label{Hc}
u(\lambda) :=  \tau_n (\hat{\tau}_n -\hat{\tau}_n^{-1})[\emph{m}_0(\lambda)]_{\tau^2}+ \tau_0 (\hat{\tau}_0 -\hat{\tau}_0^{-1})[\emph{m}_m(\lambda)]_{\tau^2} .
\end{equation}
\end{subequations}
\end{proposition}
The  details of the computation leading to this explicit formula for the restriction of the action of $L$ (from Proposition \ref{laplacian:prp}) to $\mathcal{C}(\Lambda_{n,m})$ are provided in  Subsection \ref{H:prf} (below). It involves a special instance  of Macdonald's
product formula for the generalized Poincar\'e series with distinct parameters \cite{mac:poincare}---pertaining to a stabilizer subgroup of $W$---that is recalled in Appendix \ref{appC}.

If the inner product structure of the Hilbert space is ignored (so $l^2(\Lambda_{n,m},\Delta_{n,m})\cong \mathcal{C}(\Lambda_{n,m})$), 
then the actions of the Hamiltonians in Propositions \ref{Hnm:prp} and \ref{H:prp} formally coincide upon identifying the  $q$-boson parameters $q$, $a_\pm$ and $\hat{a}_\pm$ with  the parameters stemming from the double affine Hecke algebra as follows:
\begin{equation}\label{parameters}
\boxed{q=\tau^2,\quad a_+=\tau_0\hat{\tau}_0, \quad \hat{a}_+=-\tau_0\hat{\tau}_0^{-1},\quad 
a_-=\tau_n\hat{\tau}_n, \quad \hat{a}_-=-\tau_n\hat{\tau}_n^{-1} }
\end{equation}
(and $c_\pm=a_\pm\hat{a}_\pm$, $g_\pm=a_\pm+\hat{a}_\pm$). The subsequent overall conclusion concerning the integrability of the $q$-boson Hamiltonian is now immediate.

\begin{theorem}[Quantum Integrals]\label{qint:thm}
For parameters of the form in Eq. \eqref{parameters},
the operators $H_1,\ldots ,H_n$ \eqref{q-boson-integrals} provide $n$ commuting quantum integrals for the $n$-particle open $q$-boson Hamiltonian $H$ ($=H_1$).
\end{theorem}

At this point it is not yet necessary to impose any additional reality conditions on the parameters:  our Hecke-algebraic construction holds for any $\tau_0,\hat{\tau}_0,\tau,\tau_n,\hat{\tau}_n\in\mathbb{C}^*$ not equal to a root of unity.

\subsection{Proof of Proposition \ref{H:prp}}\label{H:prf}
From the explicit action of $L$ in Proposition \ref{laplacian:prp}, it is readily seen that the action of
$H=\Pi\, L\, \Pi^{-1}$ on $f\in\mathcal{C}(\Lambda_{n,m})$
is of the form
\begin{equation*}
(H f)(\lambda)=
u(\lambda) f (\lambda)+ \sum_{\substack{\nu\in W_0 e_1\\ \lambda+\nu\in \Lambda_{n,m} }} v_{\nu} (\lambda ) f (\lambda+\nu)\qquad
(\lambda\in\Lambda_{n,m} ),
\end{equation*}
with
\begin{align*} 
u(\lambda) &\stackrel{\text{(i)}}{=} \sum_{\nu\in W_0e_1 } d_{\lambda,\nu} =
\sum_{\substack{1 \le j \le n \\ \lambda_j=0}} d_{\lambda,-e_j}
+  \sum_{\substack{1 \le j \le n \\ \lambda_j=m}} d_{\lambda,e_j}\\
&=   \tau_n(\hat{\tau}_n - \hat{\tau}_n^{-1}) \sum_{j=n-m_0(\lambda)+1}^n \tau^{2(n-j)}
+
 \tau_0(\hat{\tau}_0 - \hat{\tau}_0^{-1}) \sum_{j=1}^{\text{m}_m(\lambda)} \tau^{2(j-1)}\\
 &= \tau_n (\hat{\tau}_n -\hat{\tau}_n^{-1})[\emph{m}_0(\lambda)]_{\tau^2}+ \tau_0 (\hat{\tau}_0 -\hat{\tau}_0^{-1})[\emph{m}_m(\lambda)]_{\tau^2} 
 \end{align*}
and
\begin{align}
v_{\nu} (\lambda)
 =
\sum_{\substack{\eta\in W_0e_1\\ (\lambda +\eta)_+=\lambda+\nu}} \tau_{w_{\lambda+\eta}}^2
\stackrel{\text{(ii)}}{=}
\sum_{\mu\in W_{\lambda} (\lambda +\nu)} \tau_{w_\mu}^2 
=
\frac{W_{\lambda}(\tau^2,\tau_n^2,\tau_0^2)}{(W_{\lambda}\cap W_{\lambda+\nu}) (\tau^2,\tau_n^2,\tau_0^2)}, \label{vj}
\end{align}
where $ W_{\lambda}\subset W$ refers to the stabilizer subgroup $\{ w\in W\mid w\lambda =\lambda\}$ and
\begin{align*}
W_{\lambda}(\tau^2,\tau_n^2,\tau_0^2) &:=\sum_{\substack{w\in W\\ w\lambda=\lambda}} \tau_w^2, \\
(W_{\lambda} \cap W_{\lambda+\nu} )(\tau^2,\tau_n^2,\tau_0^2) &:=
\sum_{\substack{w\in W\\ w\lambda=\lambda \,\text{and}\\ w(\lambda+\nu)=\lambda+\nu}}  \tau_w^2 .
\end{align*}
Here we used that for any $\lambda\in\Lambda_{n,m}$: (i) $(\lambda \pm e_j)_+\neq\lambda$ and
(ii) 
$w_{\lambda\pm e_j}\in W_{\lambda}$.
Invoking  of Macdonald's explicit product formula \eqref{poincare-stab} for the generalized Poincar\'e series
$W_{\lambda}(\tau^2,\tau_n^2,\tau_0^2)$ ($=W_{\lambda}^{(n,m)}(\tau^2,\tau_n^2,\tau_0^2) $)
of the stabilizer subgroup $W_\lambda$ ($=W^{(n,m)}_\lambda$) with $\lambda\in\Lambda_{n,m}$,
entails that
\begin{equation*}
v_{e_j}(\lambda) =[\emph{m}_{\lambda_j}(\lambda)]_{\tau^2} 
   (1+\delta_{\lambda_j} \tau_n^2 \tau^{2(\emph{m}_0(\lambda)-1)})=v_j(\lambda)
 \end{equation*}
 if $\lambda+e_j\in\Lambda_{n,m}$ and
 \begin{equation*}
v_{-e_j}(\lambda)= [\emph{m}_{\lambda_j}(\lambda)]_{\tau^2} 
   (1+\delta_{m-\lambda_j} \tau_0^2 \tau^{2(\emph{m}_m(\lambda)-1)})=v_j(\lambda)
  \end{equation*}
    if $\lambda-e_j\in\Lambda_{n,m}$. 
More specifically, these compact expressions for $v_\nu(\lambda)$ with $\nu =e_j$ and $\nu = -e_j$ have their origin in the following observation:
if both $\lambda , \lambda +\nu\in\Lambda_{n,m}$ (for such a $\nu\in W_0e_1$), then
   the decomposition
    of $W_{\lambda} \cap W_{\lambda+\nu} $
    as a direct product of finite hyperoctahedral groups and permutation groups differs
from the corresponding
    decomposition of  $W_{\lambda}$ ($=W^{(n,m)}_\lambda$) in Eq. \eqref{direct-product} 
 by at most a single factor.  In view of Eq. \eqref{poincare-stab}, the expression in
 Eq. \eqref{vj} thus simplifies as the quotient of the Poincar\'e series corresponding to these distinct factors in the numerator and the denominator (as the Poincar\'e series stemming from all other factors in the decomposition appear common in the numerator and the denominator and therefore cancel):
 \begin{equation*}
 v_{e_j}(\lambda) =
 \begin{cases}
 \frac{S_{\text{m}_{\lambda_j}(\lambda)}(\tau^2)}{S_{\text{m}_{\lambda_j}(\lambda)-1} (\tau^2)} &\text{if}\  \lambda_j>0, \\
  \frac{W_0^{(\text{m}_0(\lambda))}(\tau^2,\tau_n^2)}{W_0^{(\text{m}_0(\lambda)-1)} (\tau^2,\tau_n^2)} &\text{if}\  \lambda_j=0 ,
 \end{cases} 
 \
 v_{-e_j}(\lambda) =
 \begin{cases}
 \frac{S_{\text{m}_{\lambda_j}(\lambda)}(\tau^2)}{S_{\text{m}_{\lambda_j}(\lambda)-1} (\tau^2)} &\text{if}\  \lambda_j<m, \\
  \frac{W_0^{(\text{m}_m(\lambda))}(\tau^2,\tau_0^2)}{W_0^{(\text{m}_m(\lambda)-1)} (\tau^2,\tau_0^2)} &\text{if}\  \lambda_j=m.
 \end{cases} 
 \end{equation*}
Upon evaluating the surviving Poincar\'e series  with the aid of
 Eqs . \eqref{poincare-W0}, \eqref{poincare-Sn}, this gives rise to the
 compact expressions for $v_{\pm e_j}(\lambda)$ displayed above.

\section{Completeness of the Bethe Ansatz associated with $\mathbb{H}$}\label{sec9}
In this section the commuting quantum integrals $H_1,\ldots ,H_n$ \eqref{q-boson-integrals} for  the $n$-particle open $q$-boson system are simultaneously diagonalized in $\mathcal{C}(\Lambda_{n,m})$ by means of a basis of Bethe Ansatz eigenfunctions that consists of hyperoctahedral Hall-Littlewood polynomials.
\emph{Throughout this section it will be assumed that  the  $q$-boson parameters $q$, $a_\pm$, $\hat{a}_\pm$ and the double affine Hecke algebra parameters
$\tau$, $\tau_\pm$, $\hat{\tau}_\pm$ are related through Eq. \eqref{parameters} and that none of these parameters lies on the unit circle.}

\subsection{Affine hyperoctahedral Hall-Littlewood functions}
For a wave vector $\boldsymbol{\xi}=(\xi_1,\ldots,\xi_n)\in \mathbb{R}^n$,  let $\mathbf{e}^{i\boldsymbol{\xi}}\in \mathcal{C}(\mathbb{Z}^n)$ denote the associated plane wave function of the form
\begin{equation*}
\mathbf{e}^{i\boldsymbol{\xi}}(\lambda ):=e^{i\lambda_1\xi_1+\cdots +i\lambda_n\xi_n}\qquad  (\lambda\in \mathbb{Z}^n).
\end{equation*}
The {\em affine  hyperoctahedral  Hall-Littlewood function} $\Phi_{\boldsymbol{\xi}}\in\mathcal{C}(\mathbb{Z}^n)$ with spectral parameter $\boldsymbol{\xi}$ is now defined as
\begin{equation}\label{HLF}
\Phi_{\boldsymbol{\xi}} := \mathcal{J}\phi_{\boldsymbol{\xi}}\quad\text{where}\quad \phi_{\boldsymbol{\xi}} := \Bigl( \sum_{w\in W_0} \tau_w I_w\Bigr) \mathbf{e}^{i{\boldsymbol{\xi}}} .
\end{equation}
Notice that $\phi_{\boldsymbol{\xi}}$ arises from the integral-reflection action
of the element $\mathbf{1}_0:=\sum_{w\in W_0} \tau_wT_w$  on the plane wave $\mathbf{e}^{i\boldsymbol{\xi}}$. This element is identified as the normalized idempotent
 corresponding to the trivial representation of the Hecke algebra of the hyperoctahedral group $W_0$:
\begin{equation}\label{idempotent}
T_j\mathbf{1}_0 =\tau_j \mathbf{1}_0 \qquad (j=1,\ldots ,n).
\end{equation}
By summing all contributions stemming from the idempotent, one arrives at an explicit formula for $\phi_{\boldsymbol{\xi}}$.

\begin{proposition}[Plane Waves Decomposition]\label{BWF:prp}
For
\begin{equation}\label{Rreg}
\boldsymbol{\xi}\in \mathbb{R}^n_{\text{reg}}:=\{ \boldsymbol{\xi} \in \mathbb{R}^n \mid \xi_j, \xi_j \pm \xi_k \not\in  \pi\mathbb{Z},\ \text{for}\  1\leq j\neq k\leq n \} ,
\end{equation}
the function $\phi_{\boldsymbol{\xi}}$ \eqref{HLF} decomposes into the following linear combination  plane waves:
\begin{subequations}
\begin{equation}\label{cf-decomposition}
  \phi_{\boldsymbol{\xi}} =\sum_{w\in W_0} C(w\boldsymbol{\xi })\mathbf{e}^{iw \boldsymbol{\xi}} ,
  \end{equation}
where
\begin{equation}\label{cfun}
\begin{split}
 C({\boldsymbol{\xi}}):=&
 \prod_{1\le j < k \le n} 
   \frac{ 1-q e^{-i(\xi_j-\xi_k)}}  { 1-e^{-i(\xi_j-\xi_k)}}   
      \frac{ 1-q e^{-i(\xi_j+\xi_k)}}  { 1-e^{-i(\xi_j+\xi_k)}}   \\
&\qquad  \times \prod_{1\le j \le n} 
   \frac{  (1-a_- e^{-i\xi_j}) (1-\hat{a}_- e^{-i\xi_j})   }{1- e^{-2i\xi_j}}  .
\end{split}
\end{equation}
\end{subequations}
\end{proposition}
This explicit plane waves expansion for $\phi_{\boldsymbol{\xi}}$ \eqref{HLF} amounts to a well-known formula in the affine Hecke algebra $\mathcal{H}$ going back to Macdonald, see e.g. \cite[Sec.~4]{mac:spherical}, \cite[Thm.~2.9]{nel-ram:kostka} and \cite[Thm.~6.9]{par:buildings}. To keep our presentation self-contained, a short elementary verification based on the integral-reflection representation is provided below in Subsection \ref{BWF:prf}.

\subsection{$W$-invariance}
The affine hyperoctahedral Hall-Littlewood function is automatically invariant with respect to the finite group action of $W_0$ because
for $j=1,\ldots ,n$:
\begin{equation}\label{W0-invariance}
\hat{T}_j\Phi_{\boldsymbol{\xi}}=\hat{T}_j\mathcal{J}\phi_{\boldsymbol{\xi}}\stackrel{\text{Eq.~\eqref{IT:relations}}}{=}
\mathcal{J}I_j\phi_{\boldsymbol{\xi}}\stackrel{\text{Eq.~\eqref{idempotent}}}{=}\mathcal{J}\tau_j\phi_{\boldsymbol{\xi}}=\tau_j\Phi_{\boldsymbol{\xi}} ,
\end{equation}
so $s_j\Phi_{\boldsymbol{\xi}}=\Phi_{\boldsymbol{\xi}}$.  In order for it to be also periodic
with respect to the action of the translations provided by the affine part of $W$,
the spectral parameter $\boldsymbol{\xi}$ must additionally satisfy an  algebraic system of Bethe-type equations.

\begin{proposition}[Bethe Equations]\label{BAE:prp}
For $\boldsymbol{\xi}\in \mathbb{R}^n_{\text{reg}}$ \eqref{Rreg} the affine hyperoctahedral Hall-Littlewood function
 $\Phi_{\boldsymbol{\xi}}$ \eqref{HLF}  belongs to the $W$-invariant subspace $\mathcal{C}(\mathbb{Z}^n)^{W}$ \eqref{W-inv-subspace:a},
\eqref{W-inv-subspace:b}, provided the spectral parameter $\boldsymbol{\xi}$ satisfies
the following algebraic system of Bethe type equations
\begin{align}\label{BAE}
 e^{2 im\xi_j}
&= \frac{ (1-a_+ e^{i\xi_j}) (1-\hat{a}_+ e^{i\xi_j})  } { ( e^{i\xi_j}-a_+ )   ( e^{i\xi_j}-\hat{a}_+ ) } 
  \frac{ (1-a_- e^{i\xi_j}) (1-\hat{a}_- e^{i\xi_j})  } { ( e^{i\xi_j}-a_- )   ( e^{i\xi_j}-\hat{a}_- ) } \\
 &\qquad \times \prod_{\substack{1\le k \le n \\ k\neq j}}
       \frac{ (1-q e^{i(\xi_j - \xi_k)}) (1-q e^{i(\xi_j + \xi_k)})   }
             { ( e^{i(\xi_j - \xi_k)}-q ) ( e^{i(\xi_j+ \xi_k)}-q )   } \qquad\text{for}\quad j=1,\ldots ,n . \nonumber
\end{align}
\end{proposition}
To prove this proposition it suffices to check that the Bethe equations are equivalent to the condition that
$\hat{T}_0\Phi_{\boldsymbol{\xi}}=\tau_0\Phi_{\boldsymbol{\xi}}$ (cf. Eqs. \eqref{W-inv-subspace:b} and \eqref{W0-invariance}). The details of this verification are relegated to Subsection \ref{BAE:prf} at the end.

As a direct consequence of the trivial action of the propagation operator $\mathcal{J}$ \eqref{propagator} on functions supported on the fundamental domain $\Lambda_{n,m}$ \eqref{dominant}---together with the plane waves decomposition of $\phi_{\boldsymbol{\xi}}$ in Proposition \ref{BWF:prp} and the conditioned $W$-invariance of $\Phi_{\boldsymbol{\xi}}$ in Proposition \ref{BAE:prp}---one deduces that
the affine hyperoctahedral Hall-Littlewood function \eqref{HLF} can be written explicitly in terms of hyperoctahedral Hall-Littlewood polynomials \eqref{HLp}, \eqref{Cp} when the spectral parameter satisfies the Bethe equations in Eq. \eqref{BAE}.
 
\begin{corollary}[Bethe Wave Function]\label{AHL:cor}
For $\boldsymbol{\xi}\in \mathbb{R}^n_{\text{reg}}$ \eqref{Rreg} satisfying the Bethe equations \eqref{BAE},
one has that
\begin{equation}\label{msf-decomposition}
  \Phi_\xi (\lambda)= P_{\lambda_+}(\boldsymbol{\xi};q,a_-,\hat{a}_-)\qquad (\lambda\in\mathbb{Z}^n),
 \end{equation}
where $P_{\lambda}(\boldsymbol{\xi};q,a,\hat{a})$ is given by Eqs. \eqref{HLp}, \eqref{Cp}.
\end{corollary}

\subsection{Diagonalization}
In this subsection we will  restrict our parameters further to the domain
\begin{subequations}
\begin{equation}\label{reala}
q, a_\pm, \hat{a}_\pm \in (-1,1)\setminus \{ 0\} ,
\end{equation}
or equivalently (cf. Eq. \eqref{parameters})
\begin{equation}\label{realb}
0<|\tau_0|<|\hat{\tau}_0|<1,\quad 0<|\tau_n|<|\hat{\tau}_n|<1 \quad \text{and}\quad 0<|\tau |<1,
\end{equation}
\end{subequations}
where $\tau$ and the pairs $\tau_0,\hat{\tau}_0$ and $\tau_n,\hat{\tau}_n$  are each either \emph{real} or \emph{purely imaginary}. 

Let us recall that $\boldsymbol{\xi}_\mu$, $\mu\in\Lambda_{n,m}$ denotes the unique global minimum of the semibounded strictly convex Morse function $V_\mu (\boldsymbol{\xi})$ \eqref{q-boson-morse-a}, \eqref{q-boson-morse-b}.  We will first check that this minimum provides a solution to the Bethe equations in Proposition \ref{BAE:prp}. The crux is that the critical equation $\nabla_{\boldsymbol{\xi}} V_\mu (\boldsymbol{\xi})=0$ in Eq. \eqref{q-boson-critical-eq} coincides with the Bethe equations \eqref{BAE} up to exponentiation.

\begin{proposition}[Bethe Roots]\label{bv:prp}
The minima $\boldsymbol{\xi}_\mu$, $\mu\in\Lambda_{n,m}$ provide $\frac{(m+n)!}{m!\, n!}$ (distinct) solutions for the Bethe equations \eqref{BAE} belonging to the open fundamental alcove $\mathbb{A}$ \eqref{alcove} (and thus in particular to $\mathbb{R}^n_{\text{reg}}$ \eqref{Rreg}).
\end{proposition}

\begin{proof}
If we multiply Eq. \eqref{q-boson-critical-eq}  by the imaginary unit and exponentiate  both sides (using Eq. \eqref{q-boson-morse-b}),  then it becomes clear that for any $\mu\in\Lambda_{m,n}$ the critical point $\boldsymbol{\xi}_\mu$ provides a solution to the Bethe equations \eqref{BAE}. 
As detailed at the end of Remark \ref{critical-point:rem},  the critical points in question belong to the open fundamental alcove $\mathbb{A}$ ($\subset \mathbb{R}^n_{\text{reg}}$).  Moreover, it is also manifest from Eq. \eqref{q-boson-critical-eq} that
$\boldsymbol{\xi}_\mu\neq \boldsymbol{\xi}_{\hat{\mu}}$  if $\mu\neq \hat{\mu}$.
\end{proof}

By combining Corollary \ref{AHL:cor} and Proposition \ref{bv:prp}, it follows that  for any $\mu\in\Lambda_{n,m}$ 
the value of
$\psi_\mu:=\Pi \Phi_{\boldsymbol{\xi}_\mu} $ at $\lambda\in\Lambda_{n,m}$ is given explicitly by $P_{\lambda }(\boldsymbol{\xi}_\mu ;q,a_-,\hat{a}_-)$ \eqref{HLp}, \eqref{Cp}:
\begin{equation}\label{HLF-decomposition}
\psi_\mu (\lambda)=\Phi_{\boldsymbol{\xi}_\mu} (\lambda) = P_{\lambda }(\boldsymbol{\xi}_\mu ;q,a_-,\hat{a}_-)\qquad (\forall \mu,\lambda\in\Lambda_{n,m}) .
\end{equation}

We are now in the position to formulate the main result of this paper.

\begin{theorem}[Completeness of the Bethe Ansatz Eigenfunctions]\label{completeness:thm}
For parameters in the domain \eqref{reala}, \eqref{realb}, the hyperoctahedral Hall-Littlewood functions $\psi_\mu=\Pi \Phi_{\boldsymbol{\xi}_\mu} $, $\mu\in\Lambda_{n,m}$
constitute a basis for $\mathcal{C}(\Lambda_{n,m})$ that diagonalizes the commuting $q$-boson quantum integrals $H_1,\ldots ,H_n$ \eqref{q-boson-integrals} simultaneously:
\begin{equation}\label{eveq}
H_r \psi_\mu   =  E_r(\boldsymbol{\xi}_\mu)  \psi_\mu \qquad  (\mu\in\Lambda_{n,m},\, r=1,\ldots ,n).
\end{equation}
Here $E_r(\boldsymbol{\xi}):=E_r(e^{i\xi_1},\ldots ,e^{i\xi_n})$ refers to the elementary symmetric function in
Eq. \eqref{elementary-sf} with $X_j$ replaced by $e^{i\xi_j}$ ($j=1,\ldots ,n$).
\end{theorem}

\begin{proof}
Let us first emphasize that Eq. \eqref{HLF-decomposition} implies that for any $\mu\in\Lambda_{n,m}$ the function
 $\psi_\mu\neq 0$ as an element of $\mathcal{C}(\Lambda_{n,m})$ (cf. Remark \ref{normalization:rem}).  Moreover, by acting with
$H_r$  \eqref{q-boson-integrals} on $\psi_\mu=\Pi \Phi_{\boldsymbol{\xi}_\mu} $ it is readily confirmed that we are dealing with an eigenfunction:
\begin{equation*}
H_r \psi_\mu = \Pi L_r \Phi_{\boldsymbol{\xi}_\mu} \stackrel{\text{(i)}}{=} \Pi \mathcal{J} E_r(t)\phi_{\boldsymbol{\xi}_\mu}\stackrel{\text{(ii)}}{=}
\Pi \mathcal{J}  E_r(\boldsymbol{\xi}_\mu) \phi_{\boldsymbol{\xi}_\mu}
\stackrel{\text{(iii)}}{=}  E_r(\boldsymbol{\xi}_\mu) \psi_\mu ,
\end{equation*}
where we have used (i) Eqs.~\eqref{Lr}, \eqref{HLF}, (ii) Eq.~\eqref{Er} and Proposition \ref{BWF:prp},  and (iii) 
the fact that $E_r(\boldsymbol{\xi}_\mu) $ is a scalar and can therefore be pulled out from the left.
Since the elementary symmetric functions $E_1(\boldsymbol{\xi}), \ldots ,E_n(\boldsymbol{\xi})$  separate the points of $\mathbb{A}$ \eqref{alcove}, we see that for any $\mu,\hat{\mu}\in\Lambda_{n,m}$:
\begin{equation*}
\mu\neq\hat{\mu}\stackrel{\text{Prp. \ref{bv:prp}}}{\Longrightarrow} \boldsymbol{\xi}_\mu\neq \boldsymbol{\xi}_{\hat{\mu}}\Longrightarrow E_r(\boldsymbol{\xi}_\mu)\neq E_r(\boldsymbol{\xi}_{\hat{\mu}})
\end{equation*}
for some  $r\in \{ 1,\ldots ,n\}$. This nondegeneracy ensures that the $\frac{(m+n)!}{m!\, n!}$ eigenfunctions $\psi_\mu$, $\mu\in\Lambda_{n,m}$ must be linearly independent, so they indeed provide a basis for $\mathcal{C}(\Lambda_{n,m})$.
\end{proof}

By specializing to the simplest quantum integral for $r=1$, Theorem \ref{completeness:thm} gives rise to Theorem \ref{diagonal:thm}. From the explicit formula for the operator at issue in this particular case  (cf. Propositions \ref{Hnm:prp} and \ref{H:prp}), it is manifest that we may analytically continue
the corresponding eigenvalue equation in Eqs. \eqref{EV-eq:a}, \eqref{EV-eq:b} to the slightly larger parameter domain in Eq. \eqref{domain} (cf.   Remark \ref{analyticity:rem}).

\subsection{Proof of Proposition \ref{BWF:prp}}\label{BWF:prf}
It is seen from the explicit action in Eqs.  \eqref{Ij:a}--\eqref{Ij:c} that
for any $\boldsymbol{\xi}\in \mathbb{R}^n_{reg}$ \eqref{Rreg}  and $j=1,\dots,n$:
\begin{equation}\label{I-plane}
\begin{split}\tau_j I_j \mathbf{e}^{i\boldsymbol{\xi }} &= \text{b}_j(s_j\boldsymbol{\xi}) \mathbf{e}^{i\boldsymbol{\xi}}  +\text{c}_j(s_j\boldsymbol{\xi}) \mathbf{e}^{i s_j \boldsymbol{\xi}} ,\\
&= \text{b}_j(-\boldsymbol{\xi}) \mathbf{e}^{i\boldsymbol{\xi}}  +\text{c}_j(-\boldsymbol{\xi}) \mathbf{e}^{i s_j \boldsymbol{\xi}} ,
\end{split}
\end{equation}
where
\begin{align*}
\text{b}_j(\boldsymbol{\xi}) &= \tau_j^2-\text{c}_j(\boldsymbol{\xi})= \text{c}_j(-\boldsymbol{\xi})-1,    \\
\text{c}_j(\boldsymbol{\xi}) &= \begin{cases}     \frac{1-q e^{-i(\xi_j-\xi_{j+1})}}{1-e^{-i(\xi_j-\xi_{j+1})}}      &\text{if} \  j=1,\ldots ,n-1, \\
 \frac{(1-a_- e^{-i\xi_n})(1-\hat{a}_- e^{-i\xi_n})}{(1-e^{-2i\xi_n})}   &\text{if} \  j=n.
 \end{cases}
\end{align*}
Hence  $\phi_{\boldsymbol{\xi}}$ \eqref{HLF} decomposes in this situation as a linear combination of plane waves of the form
\begin{equation}\label{pw:exp}
  \phi_{\boldsymbol{\xi}}=\sum_{w\in W_0} C_w(\boldsymbol{\xi})\mathbf{e}^{iw\boldsymbol{\xi }}
\end{equation}
for certain \emph{unique} coefficients $C_w(\boldsymbol{\xi})\in\mathbb{C}$. Here we use that for $\boldsymbol{\xi}\in\mathbb{R}^n_{\text{reg}}$ the plane waves $\mathbf{e}^{iw\boldsymbol{\xi}}$, $w\in W_0$ are linearly independent in $\mathcal{C}(\mathbb{Z}^n)$, because the corresponding wave vectors $w\boldsymbol{\xi}$, $w\in W_0$ are \emph{distinct} modulo $2\pi\mathbb{Z}^n$ (as the stabilizer of $\boldsymbol{\xi}\in\mathbb{R}^n_{\text{reg}}$ for the action of $W$ in Eq. \eqref{W-action} with $c=\pi$ is trivial).

Let us first compute the coefficient $C_{w}(\boldsymbol{\xi})$ for $w=w_0$, where $w_0$ refers to
the longest element of $W_0$ (so $w_0\boldsymbol{\xi}=-\boldsymbol{\xi}$). To this end we define
\begin{subequations}
\begin{equation}\label{rj}
r_j:=  s_js_{j+1}\cdots s_{n-1}  s_n s_{n-1}\cdots s_{j+1}s_j\quad \text{for}\quad  j=1,\ldots ,n
\end{equation} 
(so $r_j(\xi_1,\ldots ,\xi_n)=(\xi_1,\ldots,\xi_{j-1},-\xi_j,\xi_{j+1},\ldots ,\xi_n)$), which permits to decompose $w_0$ in terms of the following reduced
expression:
\begin{equation}
w_0= r_1 r_2   \cdots  r_n
\end{equation}
(of length $n^2$). From Eq. \eqref{I-plane} it is immediate that for any reduced expression $w=s_{j_\ell} \cdots s_{j_1}\in W_0$ the action of
$\tau_wI_w$ on $ \mathbf{e}^{i\boldsymbol{\xi }}$ is of the form
\begin{equation}\label{lo}
\tau_w I_w \mathbf{e}^{i\boldsymbol{\xi }} =\Bigl( \prod_{1\leq k\leq \ell} c_{j_k}(s_{j_k}\cdots s_{j_2}s_{j_1}\boldsymbol{\xi}) \Bigr) \mathbf{e}^{i w\boldsymbol{\xi }} \quad +\text{l.o.},
\end{equation}
\end{subequations}
where l.o. stands for a linear combinations of plane waves $\mathbf{e}^{i v\boldsymbol{\xi }}$ with $v<w$ in the Bruhat partial order on $W_0$ (cf. Subsection \ref{inv:prf}). With the aid of Eqs. \eqref{rj}, \eqref{lo}, one computes that
\begin{eqnarray*}
\lefteqn{\tau_{r_j} I_{r_j}  \mathbf{e}^{i\boldsymbol{\xi }}=} && \\ && \mathbf{e}^{i r_j\boldsymbol{\xi }}
\left(  \frac{  (1-a_- e^{i\xi_j}) (1-\hat{a}_- e^{i\xi_j})   }{1- e^{2i\xi_j}} 
\prod_{j < k \le n} 
   \frac{ 1-q e^{i(\xi_j-\xi_k)}}  { 1-e^{i(\xi_j-\xi_k)}}   
      \frac{ 1-q e^{i(\xi_j+\xi_k)}}  { 1-e^{i(\xi_j+\xi_k)}}  \right)  \ +\text{l.o.}
\end{eqnarray*}
Hence, we deduce that the leading coefficient (with respect to the Bruhat order) in the plane waves expansion \eqref{pw:exp} is given explicitly by
$C_{w_0}(\boldsymbol{\xi})=C(w_0\boldsymbol{\xi})= C(-\boldsymbol{\xi})$ with $C(\cdot)$ as in Eq. \eqref{cfun}.

To compute the remaining coefficients $C_w(\boldsymbol{\xi})$ for $w\neq w_0$, we note that $I_j \phi_{\boldsymbol{\xi}}=\tau_j \phi_{\boldsymbol{\xi}}$ for $j=1,\ldots,n$ (because of Eq. \eqref{idempotent}). It thus follows---by Eq. \eqref{I-plane} and the linear independence of the plane waves---that for $\boldsymbol{\xi}\in\mathbb{R}^n_{\text{reg}}$:
\begin{equation*}
C_{s_jw}(\boldsymbol{\xi}) \text{c}_j(w\boldsymbol{\xi})= C_{w}(\boldsymbol{\xi}) \text{c}_j(-w\boldsymbol{\xi})
 \quad\text{for\ all}\ 
 w\in W_0,\  j\in\{1,\dots, n\} .
\end{equation*}
Moreover, it is readily seen that $C(w\boldsymbol{\xi})$ also satisfies  this same recurrence relation, since we deduce from
the explicit product formula in Eq \eqref{cfun}  that for  $\boldsymbol{\xi}\in\mathbb{R}^n_{\text{reg}}$:
\begin{equation*}
C (s_j\boldsymbol{\xi})\text{c}_j(\boldsymbol{\xi})=
C(\boldsymbol{\xi}) \text{c}_j(-\boldsymbol{\xi})
 \quad\text{for\ all}\ 
  j\in\{1,\dots, n\} .
\end{equation*}
The upshot is that $C_{w}(\boldsymbol{\xi}) =C(w\boldsymbol{\xi})$ for all $w\in W_0$ and any $\boldsymbol{\xi}\in\mathbb{R}^n_{\text{reg}}$  (by downward  induction with respect to the Bruhat order starting from the initial condition
$C_{w_0}(\boldsymbol{\xi}) =C(w_0\boldsymbol{\xi})$ while using that  $c_j(\boldsymbol{\xi})\neq 0$).

\subsection{Proof of Proposition \ref{BAE:prp}}\label{BAE:prf}
The affine Hall-Littlewood function belongs to $\mathcal{C}(\mathbb{Z}^n)^W$
provided $\hat{T}_0\Phi_{\boldsymbol{\xi}}=\tau_0\Phi_{\boldsymbol{\xi}}$ (cf. Eqs. \eqref{W-inv-subspace:b} and \eqref{W0-invariance}), or equivalently
 $I_0\phi_{\boldsymbol{\xi}} =\tau_0\phi_{\boldsymbol{\xi}}$ (cf. Proposition \ref{IT:prp}).
Since for any $\boldsymbol{\xi}\in \mathbb{R}^n_{\text{reg}}$ \eqref{Rreg}, the action of $I_0$ \eqref{Ij:a}--\eqref{Ij:c} on $\mathbf{e}^{i\boldsymbol{\xi}}$ is determined by
\begin{equation*}
\tau_0 I_0 \mathbf{e}^{i\boldsymbol{\xi }} =  \text{b}_0(-\boldsymbol{\xi}) \mathbf{e}^{i\boldsymbol{\xi}}  +\text{c}_0(-\boldsymbol{\xi}) s_0 \mathbf{e}^{i  \boldsymbol{\xi}} ,
\end{equation*}
where $s_0 \mathbf{e}^{i  \boldsymbol{\xi}}= e^{2im\xi_1}\mathbf{e}^{i s_0^\prime\boldsymbol{\xi}}$ and
\begin{align*}
\text{b}_0(\boldsymbol{\xi}) &=  \tau_0^2-\text{c}_0(\boldsymbol{\xi})=\text{c}_0(-\boldsymbol{\xi})-1  ,    \\
\text{c}_0(\boldsymbol{\xi}) &=  \frac{(1-a_+  e^{i\xi_1})(1-\hat{a}_+ e^{i\xi_1})}{1-e^{2i\xi_1}} ,
\end{align*}
we see from the plane waves decomposition in Eqs. \eqref{cf-decomposition}, \eqref{cfun} that in this situation
\begin{align*}
\tau_0 I_0 \phi_{\boldsymbol{\xi }} =
&\sum_{w\in W_0}   \text{b}_0(-w\boldsymbol{\xi} )  C(w\boldsymbol{\xi}) \mathbf{e}^{iw\boldsymbol{\xi}}  \\
  & + \sum_{w\in W_0} \text{c}_0(w\boldsymbol{\xi} )  C(s_0' w\boldsymbol{\xi}) e^{ -2i m (w\boldsymbol{\xi})_1}  \mathbf{e}^{iw\boldsymbol{\xi}} .
\end{align*}
If we compare this expression with the corresponding plane waves decomposition of $\tau_0^2 \phi_{\boldsymbol{\xi}}$---while recalling the linear independence
of the plane waves $\mathbf{e}^{iw\boldsymbol{\xi}}$, $w\in W_0$ for $\boldsymbol{\xi}\in \mathbb{R}^n_{\text{reg}}$ (cf. Subsection \ref{BWF:prf})---then
it follows that the $W$-invariance of the affine Hall-Littlewood polynomial is guaranteed for $\boldsymbol{\xi}\in \mathbb{R}^n_{\text{reg}}$ provided
\begin{equation}\label{Bethe}
\begin{split}
e^{ 2im (w\boldsymbol{\xi})_1} &=
\frac{C(s'_0 w\boldsymbol{\xi})}{C(w\boldsymbol{\xi})} \frac{\text{c}_0 (w\boldsymbol{\xi}) } {\tau_0^2 -  \text{b}_0 (-w\boldsymbol{\xi}) } 
=\frac{C(s'_0 w\boldsymbol{\xi})}{C(w\boldsymbol{\xi})} \frac{\text{c}_0 (w\boldsymbol{\xi}) } {\text{c}_0 (-w\boldsymbol{\xi}) } \\
&=-\frac{C(s'_0 w\boldsymbol{\xi})}{C(w\boldsymbol{\xi})}
 \frac{ (1-a_+ e^{i(w\boldsymbol{\xi})_1}) (1-\hat{a}_+ e^{i(w\boldsymbol{\xi})_1})  } { (e^{i(w\boldsymbol{\xi})_1}-a_+)   (e^{i(w\boldsymbol{\xi})_1}-\hat{a}_+) } 
\qquad \forall w\in W_0.
\end{split}
\end{equation}
By substituting the explicit product formula for $C (\cdot )$ from Eq. \eqref{cfun} and canceling the
common factors in the numerator and denominator, we see that
\begin{align*}
\frac{C(s'_0\boldsymbol{\xi})}{C(\boldsymbol{\xi})}=
-
\frac{ (1-a_- e^{i\xi_1}) (1-\hat{a}_- e^{i\xi_1})  } { (e^{i\xi_1}-a_-)   ( e^{i\xi_1}-\hat{a}_-) } 
 \prod_{1 < j  \leq n}
       \frac{ (1- q e^{i(\xi_1 - \xi_j)}) (1-q e^{i(\xi_1 + \xi_j)})   }
             { (e^{i(\xi_1 - \xi_j)}-q) ( e^{i(\xi_1 + \xi_j)}-q) } ,
\end{align*}
which confirms that Eq. \eqref{Bethe} amounts to the Bethe equations stated in Eq. \eqref{BAE}.

\appendix

\section{The basic representation and Poincar\'e-Birkhoff-Witt property at critical level}\label{appA}
In this appendix Propositions \ref{pbw:prp} and \ref{pol-rep:prp} are proven. To this end we adapt corresponding results of Noumi \cite{nou:macdonald} and Sahi 
\cite{sah:nonsymmetric} to the critical level q=1, with the aid of techniques from \cite[\text{Sec. 3}]{obl:double},
\cite[\text{Sec. 4.3}]{ems-opd-sto:trigonometric}, 
\cite[\text{Ch. 2.1}]{geh:properties} and \cite[\text{Ch. 4.3}]{mac:affine}.

Let $\mathcal{A}$ denote the quotient field of the algebra
$\mathbb{C}[X,Y]$ of Laurent polynomials in the independent  indeterminates
$X_1,\dots,X_n$ and $Y_1,\dots, Y_n$. Our action of $W_0$ on $\mathbb{C}[X]$ is lifted to $\mathcal{A}$ via the field homomorphism determined by the assignment
\begin{equation*}
w(X^\lambda Y^\mu) := X^{w\lambda} Y^{w\mu}\qquad (\text{for}\ w\in W_0 \ \text{and}\ \lambda,\mu\in \mathbb{Z}^n),
\end{equation*}
where (recall) $X^\lambda = X_1^{\lambda_1} \cdots X_n^{\lambda_n}$
and $Y^\mu= Y_1^{\mu_1} \cdots Y_n^{\mu_n}$. We write $f^w$ for the result of the action of  $w\in W_0$  on $f\in \mathcal A$. 

\begin{definition} The skew-group algebra (or smash product) $\mathcal A * W_0$ is the associative unital complex algebra  characterized by the following three properties:
\begin{itemize}
\item[(i)] $\mathcal A* W_0$ contains $\mathcal A$ and the group algebra $\mathbb{C}[W_0]$ as subalgebras,
\item[(ii)] the multiplication map defines an isomorphism of $\mathbb{C}$-vector spaces
\begin{equation*}
\mathcal A \otimes_\mathbb{C} \mathbb{C}[W_0] \to \mathcal A *W_0,
\end{equation*}
\item[(iii)] and we have the cross relations
\begin{equation*} 
(fv)(gw) = f g^v vw, \quad (\forall f,g\in \mathcal A\ \text{and}\  \forall v,w\in W_0 ).
 \end{equation*}
\end{itemize}
\end{definition}
We write $\mathcal A_X*W_0$ for the subalgebra of $\mathcal A*W_0$ determined by the smash product of $\mathbb{C}[W_0]$ and
the subfield $\mathcal A_X\subset \mathcal{A}$ generated by $\mathbb{C}[X]\subset \mathbb{C}[X,Y]$.

Since $W\cong W_0\ltimes (2c\mathbb{Z})^n\subset W_0\ltimes \mathbb{Z}^n$ (cf. Subsection \ref{sec7.1}), it is clear from the definitions that
we can identity the group algebra $\mathbb{C}[W]$ as a subalgebra of $\mathcal A*W_0$ via an injective homomorphism determined by the assignment
$t_\lambda v\to Y^\lambda v$ ($\lambda\in\mathbb{Z}^n$, $v\in W_0$). In particular: $s_0 =t_{2c e_1}s_0^\prime \to Y_1^{2c}s_0^\prime$,
cf. Remark \ref{H-relations:rem}.
With this identification, the Demazure-Lusztig operators $\check{T}_j$  \eqref{Tj:a}--\eqref{Tj:b} 
 will now be interpreted as elements of $\mathcal A* W_0 $.

Let $\check{T}_j^\prime $ be given by $\check{T}_j $ \eqref{Tj:a}--\eqref{Tj:b} with $s_j$  replaced by $s_j^\prime$. (So $\check{T}_j^\prime=\check{T}_j$ unless $j=0$.) The following lemma provides a representation of $\mathbb{H}$ that can be regarded as a
$\text{q}\to 1$ degeneration of the well-known polynomial representation of the $C^\vee C$ double affine Hecke algebra due to Noumi \cite{nou:macdonald} (at the level of the affine Hecke algebra) and Sahi 
 \cite{sah:nonsymmetric} (at the level of the double affine Hecke algebra), cf. also \cite{sto:difference}.

\begin{lemma}[Noumi-Sahi Representation at q=1]\label{noumi-sahi-rep:lem}
The assignment
$T_j \mapsto  \check{T}_j^\prime$ ($j=0,\dots, n$), $X_j \mapsto X_j$ ($j=1,\dots, n$)
extends (uniquely) to an algebra homomorphism  $\mathbb{H}\to \mathcal{A}_X*\mathbb{C}[W_0]$. 
\end{lemma}
\begin{proof}
This follows from \cite[\text{Thm.~3.1}]{sah:nonsymmetric} in the limit $\text{q}\to 1$.
\end{proof}

We now modify the representation of Lemma \ref{noumi-sahi-rep:lem} by replacing $ \check{T}_0^\prime$ with $ \check{T}_0$ (cf. \cite[\text{Thm. 4.11}]{ems-opd-sto:trigonometric}).

\begin{proposition}[Basic Representation at q=1]\label{basic-rep:prp}
The assignment
$T_j \mapsto  \check{T}_j$ ($j=0,\dots, n$), $X_j \mapsto X_j$ ($j=1,\dots, n$)
uniquely extends to an algebra homomorphism $\mathbb{H} \to \mathcal A * W_0$.
\end{proposition}

\begin{proof}
To verify the statement, we have to check that the Demazure-Lusztig operators $\check{T}_0,\ldots ,\check{T}_n$ \eqref{Tj:a}--\eqref{Tj:b} 
satisfy the corresponding relations of the form in Eqs. \eqref{quadratic-relations}, \eqref{braid-relations} and \eqref{cross-relations}.  For all relations that do not involve
$\check{T}_0$, the asserted equalities are immediate from  Lemma \ref{noumi-sahi-rep:lem}.
So it suffices to consider only the remaining relations that \emph{do} involve $\check{T}_0$:
\begin{equation}\label{relations-T0}
\begin{split}
\check{T}_0 - \check{T}_0^{-1} &=\tau_0-\tau_0^{-1} ,\\
\check{T}_0 \check{T}_1   \check{T}_0 \check{T}_1 
     &=  \check{T}_1   \check{T}_0 \check{T}_1  \check{T}_0,\\
\check{T}_0 \check{T}_j    &=   \check{T}_j \check{T}_0 , \quad 1< j \le n,\\
\check{T}_0 X_1 -  \check{T}_0^{-1}X_1^{-1} &= \hat{\tau}_0^{-1}-\hat{\tau}_0 ,\\
\check{T}_0  X_j &= X_j \check{T}_0,\quad 1< j \le n.
\end{split}
\end{equation}
Given a fixed $k\in\{1,\dots, n\}$, let us pick a $\mu\in\mathbb{Z}^n$ such that $a_0(\mu)=a_k(\mu)=0$.
Then
\begin{equation*}
\check{T}_0^\prime = Y^{-\mu} \check{T}_0 Y^{\mu}
\quad \text{ and }\quad
 \check{T}_k = Y^{-\mu} \check{T}_k Y^{\mu} .
\end{equation*}
By conjugating a relation in Eq. \eqref{relations-T0} with $Y^{-\mu}$, the element $\check{T}_0$ gets replaced by $\check{T}_0^\prime$ and the element $\check{T}_k$ is left unchanged. The upshot is that the relations in Eq. \eqref{relations-T0} follow from those with $\check{T}_0$ replaced by $\check{T}_0^\prime$, which hold in turn by Lemma \ref{noumi-sahi-rep:lem}.
\end{proof}

Let us denote the image of $T_w$, $w\in W$ in the basic representation by $\check{T}_w$.
The following proposition reveals that the basic representation in Proposition \ref{basic-rep:prp} is faithful.

\begin{proposition}[Poincar\'e-Birkhoff-Witt Property]\label{pbw-basic:prp} The elements $X^\mu  \check{T}_w$ (or alternatively $ \check{T}_w X^\mu$),
with $\mu\in\mathbb{Z}^n$ and $w\in W$, are linearly independent in $\mathcal A*W_0$ over $\mathbb{C}$. 
\end{proposition}

\begin{proof}
We mimic Macdonald's proof  in \cite[(4.3.11)]{mac:affine}  of a corresponding statement for $\text{q}\neq 1$. Only the case of the elements  $X^\mu  \check{T}_w$  will be considered here, as the proof for the elements $  \check{T}_w X^\mu$ is similar.
Since $\check{T}_j= b_j(X)+c_j(X) s_j$ with
\begin{equation*}
c_j(X) 
 =   \tau_j - b_j(X)  
=    \tau_j^{-1} \frac{  (1-\tau_j \hat{\tau}_j X^{-\alpha_j}) (1+\tau_j\hat{\tau}_j^{-1} X^{-\alpha_j})   }{1-X^{-2\alpha_j}}  
\end{equation*}
($j=0,\ldots ,n$), we have that for any reduced expression  $w=s_{j_1}\cdots s_{j_\ell}\in W$:
\begin{equation}\label{triangular-basic}
\check{T}_w 
=
\sum_{v\le w} f_{vw}(X) v, 
\end{equation}
where  $\le$ refers to the Bruhat partial order on $W$ (cf. Subsection \ref{inv:prf}), and the coefficients $f_{vw}(X)$ belong to $ \mathcal{A}_X$ with
\begin{equation*}
f_{ww}(X) = c_{j_1}(X)\cdots c_{j_\ell}(X)\neq 0. 
\end{equation*}

The linear dependence of the elements $X^\mu  \check{T}_w$  would imply the existence of a nonempty finite subset $V\subset W$ such that
\begin{equation*}
\sum_{w\in V} g_w(X) \check{T}_w =0 
\end{equation*}
for certain coefficients $g_w(X)\in \mathbb{C}[X]\setminus \{ 0\}$.  Combined with Eq. \eqref{triangular-basic} this implies that 
\begin{equation}\label{dependent-basic}
\sum_{\substack{w\in V, v\in W\\ v\leq w}} g_w(X) f_{vw}(X) v=0.
\end{equation}
Since the elements $Y^\lambda w$ with $\lambda\in \mathbb{Z}^n$  and  $w\in W_0$ are (by definition) linearly independent in the algebra $\mathcal A *W_0$ when considered as a vector space over the 
field $\mathcal A_X$, it follows from Eq. \eqref{dependent-basic} that for any $v\in W$ that is dominated in the Bruhat order by  some element(s) of $V$:
\begin{equation*}
h_v(X):=\sum_{\substack{w\in V\\ w\ge v}} g_w(X) f_{vw}(X)=0.
\end{equation*}
Upon picking $v$ to be a maximal element of $V$ it is seen that in this situation $h_v(X)=g_v(X)f_{vv}(X)=0$, which contradicts the assumption that $g_v(X)\neq 0$ (because $f_{vv}(X)\neq 0$).
\end{proof}

Proposition \ref{pbw:prp} is immediate from Propositions \ref{basic-rep:prp} and \ref{pbw-basic:prp}.
Proposition \ref{pol-rep:prp} follows from Proposition \ref{basic-rep:prp} upon identifying the $X$ and $Y$ indeterminates:
$Y_j=X_j$, $j=1,\ldots ,n$. With this identification the basic representation passes over into the polynomial representation (which is no longer faithful).

\section{Affine intertwining relations}\label{appB}
In this appendix it is shown that for
any 
\begin{equation*}
f\in\mathcal{C}(\mathbb Z^n),\
\lambda\in \mathbb Z^n\ \text{and}\ \nu\in\{ \pm e_j\mid j=1,\ldots ,n\} 
\end{equation*}
one has that
\begin{equation}\label{intertwining-property}
\tau^{-1}_{w_\lambda}  (I_{w_\lambda} g)(w_\lambda(\lambda+\nu)) =
\tau_{w_{w_\lambda (\lambda+\nu)}}^{2}  f(\lambda+\nu) + d_{\lambda_+,w'_\lambda\nu} f(\lambda),
\end{equation}
where $g:=\mathcal{J}^{-1}f$ and $d_{\lambda,\nu}$ is given by Eq. \eqref{Lb}.
This affine intertwining relation lies at the basis of the proof for Proposition \ref{laplacian:prp}.

To infer Eq. \eqref{intertwining-property}, let us first provide a reduced expression for the group element $w_{\lambda +\nu}$ with  $\lambda\in \Lambda_{n,m}$ and $\nu\in\{ \pm e_j\mid j=1,\ldots ,n\}$ such that
$\lambda+\nu\not\in\Lambda_{n,m}$. 

\begin{lemma}[Reduced Expressions for $w_{\lambda\pm e_j}$]\label{l:wlambdanu} For any $\lambda\in \Lambda_{n,m}$ \eqref{dominant} and $j\in\{ 1,\ldots ,n\}$, one has that
\begin{subequations}
\begin{equation}
w_{\lambda +e_j} =
\begin{cases}
s_{\text{m}_m(\lambda)-1} \cdots s_1 s_0 s_1 \cdots s_{j-1} &\text{if}\ \lambda_j =m, \\
s_{\text{m}_m(\lambda)+\dots +\text{m}_{\lambda_j+1}(\lambda)+1} \cdots  s_ {j-2} s_{j-1}  &\text{if}\ \lambda_{j}=\lambda_{j-1} < m,
\end{cases}
\end{equation}
and
\begin{equation}
w_{\lambda -e_j} =
\begin{cases}
 s_{n-\text{m}_0(\lambda)+1} \cdots s_{n-1} s_n s_{n-1} \cdots s_j&\text{if}\ \lambda_j =0 ,\\
s_{ \text{m}_m(\lambda)+\cdots + \text{m}_{\lambda_j}(\lambda)} \cdots  s_{j+1} s_{j}  &\text{if}\ \lambda_{j}=\lambda_{j+1} > 0,
\end{cases}
\end{equation}
\end{subequations}
with all stated expressions being \emph{reduced}.
\end{lemma}
\begin{proof}
It is immediate that the stated expressions map the corresponding vector $\lambda+\nu$ into $\Lambda_{n,m}$. Moreover, since
$w_\mu = w_{s_j \mu} s_j$ with $\ell(w_\mu) = \ell(w_{s_j \mu})+1$ for any $\mu\in\mathbb{Z}^n$  and $j\in\{ 0,\ldots, n\}$ such that
$a_j(\mu)<0$,  it follows by induction on the length of $w_{\lambda+\nu}$ that the expressions in question are indeed reduced.
\end{proof}

The above reduced expressions reveal in particular that for any $\lambda\in \Lambda_{n,m}$ and $\nu\in\{ \pm e_j\mid j=1,\ldots ,n\}$ one has that
$w_{\lambda+\nu}\in W_{\lambda}$, and therefore $(\lambda+\nu)_+\neq \lambda $.

\begin{lemma}[Affine Intertwining Relations]\label{Iqm-action:lem}

For any $f\in\mathcal{C}(\mathbb{Z}^n)$, $\lambda\in \Lambda_{n,m}$ and
$\nu\in\{ \pm e_j\mid j=1,\ldots ,n\}$:
 \begin{equation*}
\tau_{w_{\lambda + \nu }} (I_{w_{\lambda + \nu}} f)((\lambda +\nu )_+)=
f(\lambda +\nu ) - d_{\lambda,\nu }  f(\lambda) .
\end{equation*}
\end{lemma}

\begin{proof}
 The proof  of the lemma is by induction on $\ell (w_{\lambda+\nu})$, starting from the trivial situation that
$\lambda+\nu\in \Lambda_{n,m}$.
When $\ell (w_{\lambda+\nu})>0$, we employ the reduced expression of Lemma  \ref{l:wlambdanu} to write
$w_{\lambda+\nu}=w_{\lambda+s_j^\prime\nu}s_j$ with $\ell (w_{\lambda+\nu})= \ell (w_{\lambda+s_j^\prime\nu})+1$.
Invoking of the induction hypothesis then yields that
\begin{align*}
\tau_{w_{\lambda+\nu}}(I_{w_{\lambda+\nu}} f)((\lambda+\nu)_+)&=
\tau_j \tau_{w_{\lambda+s'_j\nu}} (I_{w_{\lambda+s'_j\nu}} I_jf)((\lambda+s'_j\nu)_+)  \\  &=
\tau_j(I_j f) (\lambda+s'_j\nu) - \tau_j d_{\lambda,s'_j\nu} (I_jf)(\lambda)
\end{align*}
(where it was used that $s_j(\lambda+\nu)=\lambda+s_j^\prime\nu$ so
$(\lambda+s'_j\nu)_+=(\lambda+\nu)_+$). We are now in either of the following two situations:
\begin{itemize}
\item[(i)]  $1\le j \le n-1$ and $a_j(\lambda+\nu)=-1$,
\item[(ii)]  $j =0$ or $j=n$ and $a_j(\lambda+\nu)=-2$.
\end{itemize}
Since for any $f\in \mathcal{C}(\mathbb{Z}^n)$, $\mu\in \mathbb{Z}^n$ and $j\in\{ 0,\ldots, n\}$ such that $0\leq a_j(\mu)\leq 2$:
\begin{equation*}\label{Ijact}
\tau_j (I_jf)(\mu )
 =
 \begin{cases}
 \tau_j^2 f(\mu )&\text{if}\ a_j(\mu)=0\\
f(\mu-\alpha_j) &\text{if}\ a_j(\mu) =1 \\
f(\mu-2\alpha_j) -\tau_j(\hat{\tau}_j - \hat{\tau}_j^{-1})f(\mu-\alpha_j) &\text{if}\ a_j(\mu) =2
\end{cases}
\end{equation*}
(by Eqs. \eqref{Ij:a}--\eqref{Ij:c}), it is seen that
$\tau_j(I_j f) (\lambda+s'_j\nu) =f(\lambda +\nu)$ and $(I_jf)(\lambda )=\tau_j f(\lambda )$ in Case (i), while
 $\tau_j(I_j f) (\lambda+s'_j\nu) =f(\lambda+\nu)-\tau_j(\hat{\tau}_j - \hat{\tau}_j^{-1})f(\lambda )$ in Case (ii).
The lemma now follows because in Case (i) $\tau_j^2 d_{\lambda ,s_j^\prime\nu}=d_{\lambda ,\nu}$,  while
in Case (ii) $\tau_j(\hat{\tau}_j - \hat{\tau}_j^{-1})=d_{\lambda ,\nu}$ and
$d_{\lambda,s'_j\nu}=0$.
\end{proof}

We are now in the position to verify Eq. \eqref{intertwining-property}.
Since $(\lambda+\nu)_+=w_{\lambda+\nu}(\lambda+\nu)=w_{w_\lambda (\lambda+\nu)}w_\lambda (\lambda+\nu)$, one has that
\begin{align*}
 f(\lambda+\nu)=(\mathcal{J}g) (\lambda+\nu) &= \tau^{-1}_{w_{\lambda+\nu}} (I_{w_{\lambda+\nu}} g) ((\lambda+\nu)_+)\\
&=
\tau_{w_{w_\lambda (\lambda+\nu)}w_\lambda}^{-1}(I_{w_{w_\lambda (\lambda+\nu)}w_\lambda}g)((\lambda+\nu)_+)  .
\end{align*}
Moreover, because $w_\lambda (\lambda+\nu)=\lambda_+ + w_\lambda^\prime\nu$ it follows that
$w_{w_\lambda (\lambda+\nu)}\in W_{\lambda_+}$ (in view of the observation just after Lemma \ref{l:wlambdanu}), whence $\ell (w_{w_\lambda (\lambda+\nu)}w_\lambda)=\ell (w_{w_\lambda (\lambda+\nu)})+\ell (w_\lambda)$.
We may thus rewrite the above equality in a form equivalent to Eq. \eqref{intertwining-property}:
\begin{align*}
f(\lambda+\nu)&= \tau_{w_{w_\lambda (\lambda+\nu)}}^{-1} \tau_{w_\lambda} ^{-1} (I_{w_{w_\lambda (\lambda+\nu)}} I_{w_\lambda}g )((\lambda+\nu)_+)
  \\
 & =\tau_{w_{w_\lambda (\lambda+\nu)}}^{-2}  \tau_{w_\lambda} ^{-1}  \left(
(I_{w_\lambda}g)(w_\lambda (\lambda +\nu) )
 -  d_{\lambda_+,w'_\lambda \nu} 
(I_{w_\lambda} g)(\lambda_+)
\right)  \\
&= \tau_{w_{w_\lambda (\lambda+\nu)}}^{-2}\left( 
\tau_{w_\lambda} ^{-1}(I_{w_\lambda}g)(w_\lambda (\lambda +\nu) )-   d_{\lambda_+,w'_\lambda \nu}f (\lambda) \right) ,
\end{align*}
where we used Lemma \ref{Iqm-action:lem} with $f$, $\lambda$ and $\nu$ being replaced by
$I_{w_\lambda}g$, $\lambda_+$ and $w_\lambda^\prime\nu$, respectively.

\section{Hyperoctahedral Poincar\'e series with distinct parameters}\label{appC}
In this appendix Macdonald's product formula for the generalized Poincar\'e series with distinct parameters from \cite{mac:poincare} is recalled for the special instance of a stabilizer subgroup of the affine hyperoctahedral group.
For this purpose the rank $n$ of our (affine) hyperoctahedral group is made explicit here by attaching it as a superscript $(n)$.

The two-parameter generalized Poincar\'e series  $W_0^{(n)}(\tau^2,\tau_n^2)$ of the (finite) hyperoctahedral group $W_0^{(n)}\subset W^{(n)}$ 
is defined as \cite{mac:poincare}:
\begin{equation}
W_0^{(n)}(\tau^2,\tau_n^2) := \sum_{w\in W_0^{(n)}}\tau_w^2.
\end{equation}
Macdonald's celebrated product formula states that
\begin{subequations}
\begin{align}\label{poincare-W0}
W_{0}^{(n)}(\tau^2,\tau_n^2) = S_n(\tau^2)(-\tau_n^2;\tau^2)_n ,
\end{align}
where  $S_n(\tau^2)$ denotes the Poincar\'e series of the symmetric group $S_n\subset W_0^{(n)}$  generated by $s_1,\ldots ,s_{n-1}$ (which acts in the representation of Eq.  \eqref{W-action} by permuting the coordinates of $x=(x_1,\ldots ,x_n)$):
\begin{equation}\label{poincare-Sn}
S_n(\tau^2):=\sum_{w\in S_n}\tau_w^2=\sum_{w\in S_n} \tau^{2\ell (w)}
= [n]_{\tau^2} [n-1]_{\tau^2}\cdots [1]_{\tau^2} =[n]_{\tau^2}!.
\end{equation}
\end{subequations}

For any $\lambda\in\Lambda_{n,m}$, the stabilizer subgroup  $W_{\lambda}^{(n,m)}:=\{ w\in W^{(n)} \mid w\lambda =\lambda \}$ decomposes as a direct product of finite hyperoctahedral groups and symmetric groups:
\begin{equation}\label{direct-product}
W_{\lambda}^{(n,m)}
\cong
W_{0}^{(\text{m}_0(\lambda))} \times S_{\text{m}_1(\lambda)} \times \cdots \times S_{\text{m}_{m-1}(\lambda)}\times W_{0}^{(\text{m}_m(\lambda))}  .
\end{equation}
More specifically, these factors arise as the subgroups of $W_{\lambda}^{(n,m)}$ generated by the following simple reflections
\begin{align*}
W_{0}^{(\text{m}_0(\lambda))} &\cong \langle s_{\text{m}_1(\lambda)+\cdots +m_m(\lambda) +1},\ldots ,s_n   \rangle ,\\
S_{\text{m}_l (\lambda)} &\cong \langle s_{\text{m}_{l+1}(\lambda)+\cdots +\text{m}_{m}(\lambda)+1}  , \ldots , s_{\text{m}_l (\lambda)+\cdots +\text{m}_{m}(\lambda)-1} \rangle \quad (l=1,\ldots, m-1), \\
W_{0}^{(\text{m}_m(\lambda))}& \cong \langle s_0,\ldots, s_{\text{m}_{m}(\lambda)-1} \rangle .
\end{align*}
It thus follows that
\begin{align}\label{poincare-stab}
W_{\lambda}^{(n,m)}(\tau^2,\tau_0^2,\tau_n^2) &:=
\sum_{w\in W_\lambda^{(n,m)}} \tau_w^2  \\
&= W_{0}^{(\text{m}_0(\lambda))} (\tau^2,\tau_n^2) S_{\text{m}_1(\lambda)}(\tau^2) \cdots  S_{\text{m}_{m-1}(\lambda)}(\tau^2) W_{0}^{(\text{m}_m(\lambda))} (\tau^2,\tau_0^2) \nonumber \\
&=  
(-\tau_n^2;\tau^2)_{\text{m}_0(\lambda)} (-\tau_0^2;\tau^2)_{\text{m}_m(\lambda)}
\prod_{0\leq l\leq m}  [\text{m}_l(\lambda)]_{\tau^2}! . \nonumber
\end{align}

\bibliographystyle{amsplain}

\end{document}